\documentclass[lettersize,journal]{IEEEtran}

\usepackage{soul}

\usepackage{amsmath,amsfonts}
\usepackage{algorithmic}
\usepackage{algorithm}
\usepackage{array}

\usepackage[caption=false,font=footnotesize]{subfig}

\usepackage{textcomp}
\usepackage{stfloats}
\usepackage{url}
\usepackage{verbatim}
\usepackage{graphicx}
\usepackage{tikz}
\usepackage{cite}
\usepackage{amsthm}
\usepackage[normalem]{ulem}
\usepackage{xcolor}
\definecolor{revblue}{RGB}{0,70,180}
\definecolor{revred}{RGB}{180,0,0}

\newtheorem{lemma}{Lemma}
\newtheorem{remark}{Remark}
\newtheorem{example}{Example}
\newtheorem{theorem}{Theorem}
\newtheorem{corollary}{Corollary}
\newtheorem{proposition}{Proposition}
\theoremstyle{definition}
\newtheorem{definition}{Definition}

\newcommand{\var}{\mathbb{V}\mathbb{A}\mathbb{R}}
\usepackage[colorinlistoftodos,prependcaption,backgroundcolor=black!5!white,bordercolor=red]{todonotes}
\usepackage{marginnote}

\begin{document}
	
	\title{A Unified Framework for Unbiased Non-Coherent Over-the-Air Computation}
	
	\author{Martin Dahl, Zheng Chen, and Erik G. Larsson
		\thanks{This work was supported in part by ELLIIT, the Swedish Research Council (VR), and the Knut and Alice Wallenberg (KAW) Foundation. Martin Dahl is affiliated with the Wallenberg Autonomous Systems Program (WASP) Graduate School.}
		\thanks{The authors are with the Department of Electrical Engineering (ISY), Linköping University, Linköping, Sweden (e-mail:\{martin.dahl, zheng.chen, erik.g.larsson\}@liu.se).}}

	\maketitle
	
	\begin{abstract}
		Over-the-Air Computation (OAC) enables efficient data aggregation in large-scale distributed systems by exploiting the superposition property of wireless multiple-access channels. 
		In contrast to most existing studies on OAC assuming exact channel state information, we consider non-coherent OAC (NC-OAC) where the channel phase is unknown at the transmitters. 
		A three-step framework for NC-OAC with a mapping between source data and codewords is proposed: 1) Devices encode their data to non-negative codewords; 2) Devices transmit a sequence of symbols with amplitude proportional to their codewords, such that the receiver can estimate the codeword sum. Estimation of the codeword sum is studied under two scenarios of global channel amplitude knowledge: statistical or instantaneous; 3) The estimated codeword sum is decoded to the desired source data sum at the receiver. With the proposed framework, we first study prior work on NC-OAC and map these to the framework.
		Next, we define and compare the two most commonly (often implicitly) used mappings for NC-OAC: the Affine and the Augmented Affine mappings. Under the constraint of unbiased estimation, we show that with uniformly distributed data and standard channel assumptions, the Augmented Affine mapping exhibits an order of magnitude lower estimation variance than the Affine mapping with both statistical and instantaneous channel knowledge. This result is validated by extensive simulations. Finally, we propose and analyze a new mapping, which demonstrates superior performance over the previous two affine mappings.
	\end{abstract}
	

	\section{Introduction}
	Aggregation of data held by distributed
	devices over wireless networks is increasingly necessary for applications such as sensor
	networks and federated learning (FL) \cite{chen2020wireless}. The aggregation of data can be, for example, to compute a sum, a count or a majority vote. Meanwhile, radio resources are increasingly congested by a
	growing volume of data traffic, making efficient aggregation important. One approach is separating transmissions from devices in time, frequency, or space. The receiver may then decode the data from each device separately and compute the sum of the data. However, in this approach resource consumption will scale unfavorably with the number of devices, which has motivated the development of Over-the-Air Computation (OAC). By leveraging superposition in wireless channels, a weighted sum of the transmitted data naturally forms at the receiver. With appropriate modulation design at the transmitters, this enables computation of sums, or in general any nomographic function \cite{csahin2023survey, chen2023over}. In recent years, OAC has been investigated for FL \cite{yang2020federated, sery2021over}, where local gradients or model updates are periodically aggregated at a parameter server to achieve collaborative learning. FL is explained in more detail in Appendix \ref{sec:FL_mse}.

	Depending on the availability of channel knowledge, multiple variants of OAC have been studied, which can be broadly categorized into coherent and non-coherent OAC (NC-OAC). An important special case of coherent OAC is digital OAC, which is compatible with existing digital communication techniques \cite{razavikia2023channelcomp, yan2025remac, qiao2024massive}. In coherent OAC, the devices process their data and transmit the modulated signals coherently such that their signals arrive with equal phase at the receiver \cite{liu2020over,guo2021over, cao2020optimized}. Additionally, the transmitters and the receiver must compensate for their channel amplitude gains under power constraints and channel estimation errors \cite{chenyilong2023over}.  However, acquiring precise uplink channel state information (CSI) at the devices can be expensive, especially if CSI is required globally as in \cite{chenyilong2023over}. For example, the devices may transmit orthogonal uplink pilots, and receive individual feedback from the receiver on the downlink. Alternatively, if the devices are calibrated for reciprocity, a single downlink pilot is sufficient to reliably estimate the uplink channel gain at all devices. Estimating the uplink channel phase from the downlink pilot is less reliable, as oscillator phase drift, primarily caused by phase noise and carrier frequency offset \cite{tan2018mobile}, can be faster than phase drift from mobility \cite{nissel2022correctly}.
	
	For NC-OAC the channel phase is \textit{not} required; instead, the devices rely on amplitude modulation to compensate for instantaneous or expected channel fading. The power of transmitted symbols from each device adds up at the receiver, enabling estimation of a sum, or detection of a discrete count.
	Hence, compared to coherent OAC, NC-OAC significantly reduces the channel estimation overhead. In virtue of this reduced overhead, NC-OAC was the first variant of OAC to be implemented in practice \cite{kortke2014analog}, and was studied early in the literature \cite{goldenbaum2013harnessing}. However, the existing literature  on NC-OAC lacks a common framework analyzing the various NC-OAC schemes, and their associated optimal mapping design. Previous works either focus on a downstream task such as FL \cite{csahin2023distributed}, or analyze quite specialized modulation schemes \cite{csahin2024over}. \IEEEpubidadjcol
	
	\subsection{Contributions}
	Our contributions can be summarized as follows:
	\begin{itemize}
		\item \textbf{A unified and general framework} to  understand and interpret the processing chain of existing NC-OAC schemes. As explained in Section \ref{sec:ncota} and visualized by Fig. \ref{figtheframework}, the framework consists of three main steps where source data is encoded to codewords, transmitted and then decoded to the source data sum at the receiver.
		\item \textbf{Theoretical analysis} of the performance of the proposed source data sum estimator under two commonly used affine NC-OAC mappings, subject to unbiased estimation. For both mappings, we derive the estimation variance conditional on the source data, as well as the total variance assuming uniformly distributed source data. Furthermore, for a relevant special case, we show that the codeword sum estimator used in most prior work is the maximum-likelihood estimator.
		\item \textbf{A new mapping and insights on optimal mapping design:} A new \textit{Extended Affine} mapping is proposed and analyzed. With this mapping, the data are mapped onto an arbitrary number of dimensions, which is shown to improve the estimation performance compared to the two existing affine mappings. Furthermore, we show that the \textit{Unit Affine mapping} yields the minimum variance of all Affine mappings for uniformly distributed scalar source data, which has not been shown previously. In previous works, these mappings are heuristically chosen without optimality proof. Additionally, the results suggest that combining the Affine and Augmented Affine mappings can help better adapt to the distribution of source data. 
		\item \textbf{Numerical evaluation} of the three NC-OAC mappings with the proposed sum estimator for a various number of antennas, devices and sequence lengths $M,K,L$ showing that our conclusions hold numerically. The framework is also applied to FL, where our conclusions still hold. Finally, simulations with heavy-tailed source data distributions demonstrate a need of mappings robust to outliers.
	\end{itemize}
	
	\subsection{Relation to Prior Work}
	Related prior works commonly exploit the second moment of received symbols for estimating a sum, majority vote or count of source data. The mappings in these works are almost exclusively equivalent to either the \textit{Affine} or the \textit{Augmented Affine} mappings studied in Section \ref{sec:affine_mappings} without recognizing them as common patterns in NC-OAC. Table \ref{tab:mapping_applications} summarizes and maps these works to our framework where possible. In the table  ``A" stands for Affine, and ``AA" for Augmented Affine. The application scenario ``Application" is either distributed (``Dist.") or decentralized (``Decentr."), as well as for FL, a sensor network (``WSN"), a prototype (``Prot.") or distributed consensus (``Dist. Cons."). The modality of the source data ``Source" at the distributed devices is continuous ``Cont." or discrete ``Disc.". The ``Type" denotes if NC-OAC is used for summation, majority voting (MV) or a discrete count. 
	
	For multiple prior works in Table \ref{tab:mapping_applications}, various channel models and modulations are applied to transmit and receive the codewords. However, the underlying principle is the same as in our unified framework. An exception not covered by the framework is \cite{dong2020blind} which does not rely on power modulation or any of the mappings considered herein. Instead, \cite{dong2020blind} employs a gradient-descent-based method that requires substantial local processing. In Section \ref{sec:demonstration_coverage} a few of the prior works \cite{krohn2005collaborative, michelusi2023decentralized, hoque2023chirp, csahin2023over} are explicitly mapped to the framework. 
	
	\begin{table}[h!]
		\caption{Overview of previous work, based on mappings, applications, source data modality and aggregation type.}
		\centering
		\begin{tabular}{|c|c|c|c|c|c|}
			\hline
			\textbf{Year} & \textbf{Ref.} & \textbf{Mapping} & \textbf{Application} & \textbf{Source} &\textbf{Type} \\
			\hline
			2005 & \cite{krohn2005collaborative} & A & Prot. WSN & Disc. & Count \\
			2011 & \cite{goldenbaum2011analyzing} & A & WSN & Cont. & Sum\\
			2012 & \cite{jakimovski2012collective} & A & Prot. WSN & Disc. & Count\\
			2012 & \cite{sigg2012calculation} & A & Prot. WSN & Cont. & Sum \\
			2013 & \cite{goldenbaum2013harnessing} & A & WSN & Cont. & Sum\\
			2013 & \cite{goldenbaum2013robust} & A & WSN & Cont. & Sum \\
			2013 & \cite{goldenbaum2013reliable} & A & WSN & Cont. & Sum \\
			2014 & \cite{kortke2014analog} & A & Prot. WSN & Cont. & Sum \\
			2020 & \cite{gadre2020quick} & AA & Prot. WSN & Cont. & Sum \\
			2020 & \cite{dong2020blind} & Other & WSN & Cont. & Sum \\
			2021 & \cite{molinari2021max} & A & Dist. Cons. & Cont. & Sum\\
			2021 & \cite{frey2021over} & Other & WSN & Cont. & Sum \\
			2022 & \cite{adeli2022multi} & AA & FL & Disc. & MV \\
			2022 & \cite{csahin2022demonstration} & A & Prot. FL & Disc. & MV\\
			2023 & \cite{csahin2023over} & Other & FL & Cont. & Sum \\
			2023 & \cite{csahin2023distributed} & AA & FL & Disc. & MV \\
			2023 & \cite{hoque2023chirp} & AA & FL & Disc. & MV\\
			2023 & \cite{michelusi2023decentralized} & A & Decentr. FL & Cont.& Sum \\
			2023 & \cite{hellstrom2023optimal} & A & WSN & Cont. & Sum \\
			2024 & \cite{csahin2024over} & AA & WSN & Disc. & MV \\
			2024 & \cite{michelusi2024non} & AA & Decentr. FL & Cont. & MV \\
			2024 & \cite{lee2024performance} & A & WSN & Cont. & Sum \\
			2025 & \cite{deng2025robust} & AA & Dist. Cons. & Cont. & Sum \\
			\hline
		\end{tabular}
		\label{tab:mapping_applications}
	\end{table}
		
	\begin{figure*}[h]
		\centering
		\begin{tikzpicture}[node distance=2.5cm, auto, >=latex, thick, scale=1.0]
			\node at (-1,0) {$$};
			\node (xk) at (0,0) {$x_k$};
			\node at (2.5, 1.75) {\text{1. Encode}};
			\node (box1) at (3.5,0) [draw, minimum width=2.0cm, minimum height=0.5cm] {$\mathbf{w}_k=\begin{bmatrix}
					w_{k,1}\\
					\vdots\\
					w_{k,D}
				\end{bmatrix}$};
			\node (D) at (1.25,0.25) {$\mathcal{E}$};
			\node (wkn) at (5.52,0.25) {$w_{k,n}\equiv w_k$};
			\node at (7.75, 1.75) {\text{2. Transmit}};
			\node (box2) at (7.75,0) [draw, minimum width=2.0cm, minimum height=0.5cm] {$\mathbf{t}_k=\sqrt{w_k}\begin{bmatrix}
					s_{k,1}\\
					\vdots\\
					s_{k,L}
				\end{bmatrix}$};

			\draw[->] (xk) -- (2.3,0);
			\draw[->] (4.7,0) -- (6.325,0);
			\draw[->] (9.19,0) -- (11,1);
			\node () at (11.6,1) {$\mathbf{t}_kg_{1,k}$};
			\node at (11, 0.5)  {$\vdots$};
			\draw[->] (9.19,0) -- (11,0);
			\node () at (11.6,0) {$\mathbf{t}_kg_{m,k}$};
			\node at (11, -0.5) {$\vdots$};
			\draw[->] (9.19,0) -- (11,-1);
			\node () at (11.6,-1) {$\mathbf{t}_kg_{M,k}$};
			
			\draw[->] (12.2,0) -- (13,0);
			\draw[->] (13.5,1) -- (13.5,0.55);
			\node at (13.5,1.2) {$\mathbf{n}_m$};
			\node at (13.5, 1.75) {\text{3. Superposition}};
			\node (sum) at (13.5,0) [draw, circle, minimum size=1cm] {+};
			
			\draw [->] (14.5,0.75) -- (14,0.1);
			\node () at (15.2,0.75) {$\mathbf{t}_{k'}g_{m,k'}$};
			\node at (14.5, 0.1) {$\vdots$};
			\draw [->] (14.5,-0.75) -- (14,-0.1);
			\node () at (15.2,-0.75) {$\mathbf{t}_{K}g_{m,K}$};
			
			\draw [->] (13.5,-0.5) -- (13.5,-1.8);
			\node at (14.75,-2.3) {$\eta$};
			\draw[->] (14.5,-2.3) -- (14,-2.3);
			
			\node at (13.8, -1.3) {$\mathbf{r}_m$};
			
			\node (times) at (13.5,-2.3) [draw, circle, minimum size=1cm] {$\times$};
			
			\node at (13.5,-3.55) {\text{4. Receive}};
			
			\node at (12.25, -2.55) {$\widetilde{\mathbf{r}}_m$};
			\draw[->] (13,-2.3) -- (10.25, -2.3);
			\draw[->] (12.0,-1.55) -- (10.35, -2.2);
			\node at (12.25, -1.55) {$\widetilde{\mathbf{r}}_1$};
			\node at (11.75, -1.9) {$\vdots$};
			\draw[->] (12.0,-3.05) -- (10.35, -2.4);
			\node at (12.25, -3.05) {$\widetilde{\mathbf{r}}_M$};
			\node at (11.75, -2.55) {$\vdots$};
			\node at (8.75,-3.55) {\text{5. Estimate Power}};
			
			\node at (8.75, -2.3) {$\widehat{\sigma^2}=\sum_{m=1}^M\frac{\widetilde{\mathbf{r}}_m^H\widetilde{\mathbf{r}}_m}{ML}$};
			\node at (8.75,-2.3) [draw, minimum width=3.0cm, minimum height=1cm] {};
			
			\draw[->] (7.25,-2.3) -- (6.5,-2.3);
			\node at (5.5,-3.55) {\text{6. Estimate Codeword}};
			
			\node at (5.5,-2.3) {$\widehat{w} = \widehat{\sigma^2}-\eta$};
			\node at (5.5,-2.3) [draw, minimum width=2.0cm, minimum height=1cm] {};
			
			\node (box1) at (2.25,-2.3) [draw, minimum width=2.0cm, minimum height=0.5cm] {$\widehat{\mathbf{w}}=\begin{bmatrix}
					\widehat{w_{1}}\\
					\vdots\\
					\widehat{w_{D}}
				\end{bmatrix}$};
			\draw[->] (4.5,-2.3) -- (3.25,-2.3);
			\node at (4.25,-2.11) {$\widehat{w}_d$};
			\draw[->] (4.15,-1.55) -- (3.325,-2.2);
			\node at (4.45,-1.5) {$\widehat{w}_1$};
			\node at (3.85, -2.45) {$\vdots$};
			\node at (3.85, -1.95) {$\vdots$};
			\draw[->] (4.15,-3.05) -- (3.325,-2.4);
			\node at (4.45,-3.05) {$\widehat{w}_D$};
			
			\draw[->] (1.25,-2.3) -- (0.25,-2.3);
			\node at (0.75,-2.05) {$\mathcal{D}$};
			\node at (1.5,-3.55) {\text{7. Decode}};
			\node at (0,-2.3) {$\widehat{x}$};
		\end{tikzpicture}
		\caption{Overview of the non-coherent over-the-air computation framework.}
		\label{figtheframework}
	\end{figure*}
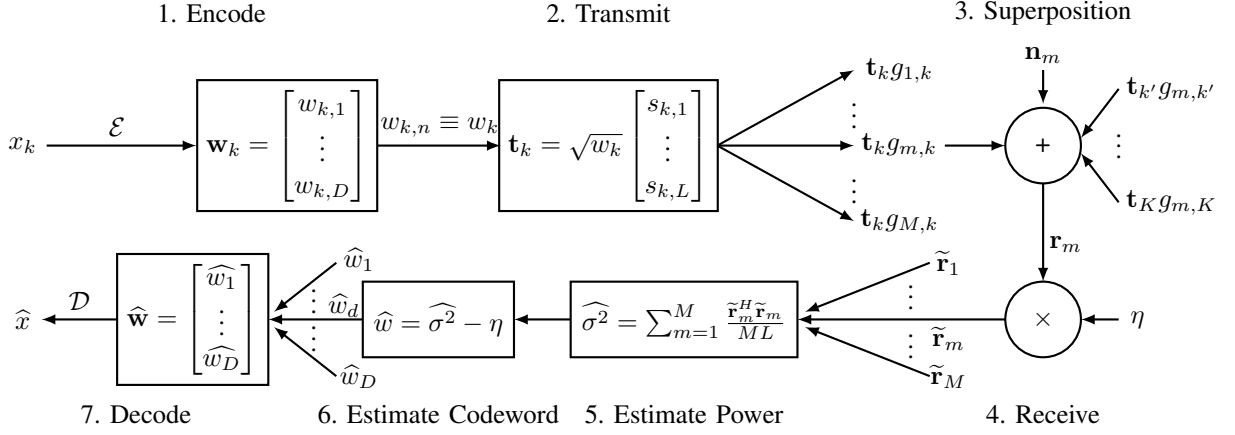

	\section{Problem Formulation}
	Consider a network of $K$ devices indexed by $k\in\{1,2,\ldots,K\}$, each of them  holding a local data sample $x_k\in[x_\text{min}, x_\text{max}]\subset \mathbb{R}$.\footnote{For high-dimensional data, the proposed framework in Fig. \ref{figtheframework} can be repeated separately for each dimension.} The goal is to estimate the sum of data from all devices:\footnote{Other nomographic functions \cite{goldenbaum2011analyzing} can be considered. The sum is chosen for simplicity and clarity.}
	\begin{equation}
		x = \sum_{k=1}^{K}x_k,
	\end{equation}
	at some receiving node, such as an access-point (AP). We consider an unbiased estimator $\widehat{x}$ given $x_k, \forall k$ minimizing its variance under some distribution of $x_k$:
	\begin{equation}\label{eq:objective_covar}
		\begin{split}
			\underset{\widehat{x}}{\text{ min }}&\var(\widehat{x}),\\
			&\text{s.t. }\mathbb{E}\left[\widehat{x}|\mathbf{x}\right] = x,
		\end{split}\tag{P1}
	\end{equation}
	where $\mathbf{x}=[x_1\ldots x_K]^\text{T}\in[x_\text{min},x_\text{max}]^K$. The constraint of unbiased estimation is motivated by the main application scenario of OAC, which is FL. As shown in Appendix \ref{sec:FL_mse}, convergence is guaranteed in mean-square under specific assumptions on the learning objective if the aggregated gradient is unbiased. Nevertheless, the mean-square-error (MSE) of an unbiased estimator is equivalent to its variance, given by Remark \ref{rem:MSE}. 
	
	\begin{remark}\label{rem:MSE}
		If $\mathbb{E}[\widehat{x}|\mathbf{x}]=x$ it follows from first principles that
		\begin{equation}
			\begin{split}
				&\var(\widehat{x}) = \mathbb{M}\mathbb{S}\mathbb{E}(\widehat{x}) + \var(x),\\
				&\text{where }\mathbb{M}\mathbb{S}\mathbb{E}(\widehat{x}) = \mathbb{E}[(\widehat{x} - x)^2],
			\end{split}
		\end{equation}
		under any distribution on $x$ with finite moments.
	\end{remark}
	\begin{proof}Appendix \ref{proof_theorem_total_variance}.
	\end{proof}
	
	\subsection{Channel Model}\label{sec:channel_model}
	Consider a block-fading multiple access channel between the $K$ single-antenna devices and the AP equipped with $M$ antennas.\footnote{The $M$ antennas can also be $M$ sub-carriers, or $M$ coherence blocks in time. In any case, it is a composite vector channel with dimension $M$.} 
	Each device $k$ transmits a length-$L$ sequence of symbols, where $t_{k,l}\in\mathbb{C}$ is the symbol transmitted by device $k$ at time instant $l\in\{1,\dots,L\}$. Each transmitted symbol is subject to an average power constraint
	$\mathbb{E}\left[|t_{k,l}|^2\right]\leq P$. Then at the $l$-th time instant the $m$-th antenna measures
	\begin{equation}\label{eq:channel_model}
		r_{m,l} = \sum_{k=1}^{K}g_{m,k}t_{k,l} + n_{m,l}\in\mathbb{C}, 
	\end{equation}
	where ${g_{m,k}\sim\mathcal{CN}(0,\beta_k)}$ is the channel coefficient between
	device $k$ and antenna $m$, independent across $m,k$ and $n_{m,l}\sim\mathcal{CN}(0,1)$ is additive noise, independent and identically distributed (i.i.d.) across $m,l$. The large-scale fading coefficient $\beta_k$ is assumed to be constant, accounting for path loss and shadowing. Furthermore,  $g_{m,k}$ is independent of  $l$ as the channel coherence time is assumed larger than $L$. For example, this channel model is relevant if orthogonal frequency-division multiplexing is used. The modulation of each transmitted symbol $t_{k,l}$ for estimation of $x$ is explained in the following sections.
	
	\subsection{Recap of Coherent Over-the-Air Computation}
	For coherent OAC, a transmitted symbol $t_{k,l}$ is modulated as follows. Assume each device knows $g_{m,k}$ through a genie, omit the subscript $m$ and consider $M=1$ for notational simplicity. Disregarding any power constraint, if $g_k$ is exactly known at device $k$, channel-inversion precoding can be applied as  $t_{k,l} = \frac{x_k}{g_k}$. Then the AP receives $r_l = \sum_{k=1}^Kx_k + n_l$ giving an unbiased least-squares estimate of $x\in\mathbb{R}$:
	\begin{equation}
		\widehat{x} = \sum_{l=1}^L\frac{\Re[r_l]}{L} = x + \sum_{l=1}^L\frac{\Re[n_l]}{L},
	\end{equation}
	highlighting the importance of knowing the channel phase $\angle g_{k}$. Here only the real dimension is used, but by including the imaginary dimension two codewords can be transmitted with one channel symbol. For this reason coherent OAC breaks down if $\angle g_{k}$ is unknown or inaccurately estimated, motivating non-coherent processing.
	
	\subsection{Non-Coherent Over-the-Air Computation}\label{sec:ncota}
	In contrast to coherent OAC, in NC-OAC the channel phase $\angle g_{m,k}$ is assumed unknown to all devices and the AP,  at most the amplitude gain $|g_{m,k}|$ or expected power gain $\beta_k$ is known. To this end, consider two scenarios of genie-aided global CSI knowledge, known exactly to all devices and the AP:

	\begin{enumerate}
		\item \textbf{Statistical CSI (SC)}: \textit{The expected power gain of the channel, $\beta_k$, is known to all, while the instantaneous CSI $g_{m,k},\forall m,k$ is unknown}.
		\item \textbf{Instantaneous CSI (IC)}: \textit{The amplitude gain of the instantaneous channel, $|g_{m,k}|, \forall k,m$,  is known to all, while the phase $\angle g_{m,k}$ is unknown}.
	\end{enumerate}
	Some prior work \cite{molinari2021max, michelusi2024non} estimate the sum CSI as a part of their NC-OAC scheme, which is either $\sum_{k=1}^K\beta_k$ or $\sum_{k=1}^K |g_{m,k}|$ and follows from knowing $\beta_k$ or $|g_{m,k}|$ individually. The power of each received symbol $r_{m,l}$, revealing the relevance of such channel knowledge, is as follows:
	\begin{equation}\label{eq:rml_second_moment}
		\mathbb{E}\left[|r_{m,l}|^2\right] = \sum_{k=1}^K\mathbb{E}\left[|g_{m,k}|^2\right]\mathbb{E}\left[|t_{k,l}|^2\right] + 1,
	\end{equation}
	under the channel model herein. The second moment $\mathbb{E}\left[|r_{m,l}|^2\right]$ yields a weighted sum which is the core of the NC-OAC framework, summarized by the flowchart in Fig. \ref{figtheframework} and the following three steps:
	\begin{enumerate}
		\item \textbf{Encode} real-valued source data to non-negative codewords, described in Section \ref{sec:affine_mappings}, corresponding to \textit{1. Encode} in Fig. \ref{figtheframework}.
		\item \textbf{Estimate} the codeword sum from received superimposed signals by modulating $\mathbb{E}\left[|t_{k,l}|^2\right]$ with the non-negative codewords, adding up linearly at the receiver side as in (\ref{eq:rml_second_moment}). This block consists of multiple intermediate steps, including \textit{2. Transmit, 3. Superposition, 4. Receive, 5. Estimate Power, 6. Estimate Codeword} in Fig. \ref{figtheframework}. Section \ref{sec:modulation} and \ref{sec:estimate_w} describe each step in detail. 
		\item \textbf{Decode} the codeword sum to the source data sum, described in Section \ref{sec:affine_mappings} and \ref{sec:estimate_x}, corresponding to step \textit{7. Decode} in Fig. \ref{figtheframework}.
	\end{enumerate}
	
	\section{Encoder and Decoder: Two Affine Mappings}\label{sec:affine_mappings}
	Without channel phase knowledge coherent precoding is impossible; instead, source data $x_k$ must be mapped onto the amplitude of each transmitted symbol $t_{k,l}$. Here multiple such mappings are introduced, encoding data $x_k$ to non-negative\footnote{$\mathbb{R}_{\geq 0}$ denotes the set of non-negative real numbers. } $D$-dimensional codewords $\mathbf{w}_k\in\mathbb{R}_{\geq 0}^D$ modulating each $t_{k,l}$. 
	
	The encoder ${\mathcal{E}: \mathbb{R}\rightarrow \mathbb{R}_{\geq 0}^D}$ at each device is associated to a decoder ${\mathcal{D}:  \mathbb{R}_{\geq 0}^D\rightarrow \mathbb{R}}$ at the AP, under the following constraints:
	
	\begin{equation}
		\begin{split}
			&\mathbf{w}_k = \mathcal{E}(x_k)\in[w_\text{min},w_\text{max}]^D,\\
			&x = \mathcal{D}(\mathbf{w})\in[Kx_\text{min}, Kx_\text{max}],
		\end{split}
	\end{equation}
	where
	\begin{equation}
		\mathbf{w}=\sum_{k=1}^K\mathbf{w}_k\in[Kw_\text{min},Kw_\text{max}]^D.
		\label{codeword_sum}
	\end{equation}
Through this structure, to estimate the sum of source data $x$ one may simply estimate the codeword sum $\mathbf{w}$. 

	\subsection{Affine Mapping (A)}\label{A}
	The \textit{Affine mapping} is an affine operator with $D=1$, mapping a scalar $x_k$ to another scalar $w_k$ by
	\begin{equation}\label{A_encoder}
		w_k = \mathcal{E}(x_k) = ax_k + b,
	\end{equation}
	and
	\begin{equation}\label{A_decoder}
		x = \mathcal{D}(w) = cw + d,
	\end{equation}
	where $a, b, c, d\in\mathbb{R}$. To satisfy (\ref{A_encoder}) and
	(\ref{A_decoder}) it must hold that
	\begin{equation}\label{eq:A_constraints}
		\begin{split}
			ca = 1\text{ and }d = -Kcb.
		\end{split}
	\end{equation}

	\begin{example}[Unit Affine mapping]\label{simple_A}
		As a simple affine mapping satisfying (\ref{eq:A_constraints}) one may choose
		\begin{equation}\label{eq:simple_affine_mapping_coeff}
			\begin{split}
				&a=\frac{1}{x_{\text{max}}-x_{\text{min}}},\text{ }b = -x_{\text{min}}a,\text{ }c = \frac{1}{a},\text{ }d = Kx_{\text{min}},
			\end{split}
		\end{equation}
		giving $w_\text{min}=0,w_\text{max}=1$. 
	\end{example}

	\subsection{Augmented Affine Mapping (AA)}\label{AA}
	The \textit{Augmented Affine mapping} is an affine mapping of a nonlinear transform with $D=2$. First, it augments $x_k$ to $\overline{\mathbf{x}}_k\in\mathbb{R}_{\geq 0}^2$ where $x_k$ is split by sign into two segments:
	\begin{equation}
		\overline{\mathbf{x}}_k = 
		\begin{bmatrix}
			(x_k)^+ \\
			(-x_k)^+ 
		\end{bmatrix},
	\end{equation}
	where $(\cdot)^+=\text{max}(\cdot,0)$. Then 
	\begin{equation}\label{AA_encoder}
		\mathbf{w}_k = \mathcal{E}(x_k) = \mathbf{A}\overline{\mathbf{x}}_k + \mathbf{b},
	\end{equation}
	and
	\begin{equation}\label{AA_decoder}
		x  = \mathcal{D}(\mathbf{w}) = \mathbf{c}^\text{T}\mathbf{w} + d,
	\end{equation}
	where $\mathbf{A}\in\mathbb{R}^{2\times 2}, b\in\mathbb{R}^2, \mathbf{c}\in\mathbb{R}^{2}, d\in\mathbb{R}$. To satisfy (\ref{AA_encoder}) and (\ref{AA_decoder}) it must hold that
	\begin{equation}\label{eq:AA_constraints}
		\begin{split}
			\mathbf{c}^\text{T}\mathbf{A} = [1\text{ }-1] \text{ and }d = -K\mathbf{c}^\text{T}\mathbf{b}.
		\end{split}
	\end{equation}
	\begin{example}[Unit Augmented Affine mapping]\label{simple_AA}
		To satisfy (\ref{eq:AA_constraints}), one may  choose for example
		\begin{equation}\label{eq:simple_augmented_affine_coeff}
			\begin{split}
				&\mathbf{A}=\begin{bmatrix}
					\frac{1}{x_{\text{max}}}&0\\
					0 & \frac{-1}{x_{\text{min}}}
				\end{bmatrix},\text{ }b = 0,\text{ }\mathbf{c} = [x_{\text{max}}\text{ } x_{\text{min}}]^\text{T}, d=0,
			\end{split}
		\end{equation}
		giving $w_\text{min}=0,w_\text{max}=1$. Note that $x_\text{min}<0$ and $x_\text{max}>0$ must hold.
	\end{example}
	
The Augmented Affine mapping distributes the dynamic range $[x_\text{min}, x_\text{max}]$ of $x_k$ over two codewords, compared to the Affine mapping.

 \subsection{Extended Affine Mapping}
 The \emph{Extended Affine} mapping is an enhanced version of the  Unit Augmented Affine mapping. 
 Here the data $x_k$ are split into an even number $N$ equally sized segments with $D=2N$ codewords. The first $N$ encode the continuous value in each segment; the last $N$ indicate which segment $x_k$ is in. However, as will become evident shortly, in practice only $2N-2$ codewords are required. For notional simplicity, consider unit source data $x_k\in[-1,1]=\bigcup_{n=1}^{N}\mathcal{X}_n$ where 
\begin{equation*}
	\mathcal{X}_n=
	\begin{cases}
		\left[\frac{2(n-1)}{N}, \frac{2n}{N}\right] & \text{ for }n\leq \frac{N}{2},\\
		\left[-\frac{2}{N}\left(n-\frac{N}{2}\right), -\frac{2}{N}\left(n-\frac{N}{2}-1\right)\right] & \text{ for }n > \frac{N}{2},
	\end{cases}
\end{equation*}
 such that there is no overlap between each $\mathcal{X}_n$ except at the interval boundaries. The mapping $\mathbf{w}_k=\mathcal{E}(x_k)$ is defined as
\begin{equation*}
	\begin{split}
		\mathcal{E}(x_k) &= [w_{k,1},w_{k,2},\dots,w_{k,N},b_{k,1},\dots,b_{k,N}]^\text{T}\in\mathbb{R}^{2N},
	\end{split}
 \end{equation*}
 where
\begin{equation}\label{eq:AA_N_encoder}
	\begin{split}
		&w_{k,n}=\begin{cases}
			b_{k,n}\left(x_k-\frac{2}{N}(n-1)\right)\frac{N}{2} & \text{ for }n\leq \frac{N}{2},\\
			b_{k,n}\left(-x_k-\frac{2}{N}\left(n-\frac{N}{2}-1\right)\right)\frac{N}{2}& \text{ for }n > \frac{N}{2},
		\end{cases}\\
		&b_{k,n} = \text{I}(x_n\in\mathcal{X}_n)\text{ for }n\in\{1,\dots,N\},
	\end{split}
\end{equation}
and $\text{I}(\cdot)$ is the indicator function. The inverse mapping is
\begin{equation}\label{eq:AA_N_decoder}
	\begin{split}
		&x = \mathcal{D}(\mathbf{w}) = \frac{2}{N}\left(\sum_{n=1}^{N/2}\left(w_{n}+K_n(n-1)\right)\right.\\
		&\left. - \sum_{n'=\frac{N+2}{2}}^{N}\left(w_{n'}+K_{n'}\left(n'-\frac{N}{2}-1\right)\right)\right),
	\end{split}
\end{equation}
 where $w_n=\sum_{k=1}^{K}w_{k,n}, \text{ }K_n=\sum_{k=1}^{K}b_{k,n}$. Since the indicators for $n=1$ and $n=N/2+1$ have no contribution, only $N-2$ counts $K_n$ are required for decoding. Note that the Unit Augmented Affine mapping is a special case of the Extended Affine when $N=2$, ignoring the two redundant indicator codewords.
 
\subsection{Mappings for Majority Voting}\label{sec:majority_voting}
Some previous NC-OAC works rely on detection of a majority vote or count of binary data $x_k\in\{-1,1\}$ \cite{krohn2005collaborative, jakimovski2012collective, adeli2022multi, csahin2022demonstration, csahin2023distributed, hoque2023chirp, csahin2024over}. In particular, our proposed Extended Affine mapping requires counting. A count is the discrete sum $x=\sum_{k=1}^{K}(x_k+1)/2\in\{0,1,\dots,K\}$, which is an important distinction to the continuous sums $x\in[0,K]$ herein, since optimally the decoder should detect a discrete value. Deriving an optimal detector is difficult as argued in Section \ref{sec:optimal_detection_MV_and_C}. As a relaxation, a count or vote can be performed through sum estimation, as herein. Majority voting is the most common special case of counting, which assuming a non-equal outcome can be expressed as
\begin{equation}\label{eq:sign_function_mv}
	x = \text{sign}\left(\sum_{k=1}^{K}x_k\right).
\end{equation}
For NC-OAC this can be formulated as an encoder and decoder similar either to the Unit Affine mapping as
\begin{equation}
	\text{MV-A: }\begin{cases}
		w_k = \mathcal{E}(x_k) = (1+x_k)/2,\\
		x = \mathcal{D}(w) = \text{sign}(w-K/2),
	\end{cases}
\end{equation}
or the Unit Augmented Affine mapping as
\begin{equation}
	\text{MV-AA: }\begin{cases}
		\mathbf{w}_k = \mathcal{E}(x_k) = [(x_k)^+\text{ }(-x_k)^+]^\text{T},\\
		x = \mathcal{D}(\mathbf{w}) = \text{sign}\left([1\text{ }-1]\textbf{w}\right).
	\end{cases}
\end{equation}
Under the effects of noise, the above decoders will never detect a tie, however, if $K$ is large or odd a tie is rare or impossible, respectively. Similarly for counting, one may select the encoder as for voting, and the decoder $\mathcal{D}(\mathbf{w})$ as the affine sum decoders (\ref{A_decoder}), (\ref{AA_encoder}) rounded to the nearest integer in $\{0,1,\dots,K\}$. 
\section{Estimation of the Codeword Sum}\label{sec:modulation}
After each device encodes its data $x_k$ to codewords $\mathbf{w}_k$, the framework proceeds to estimate the sum of the codewords $\mathbf{w}=\sum_{k=1}^K\mathbf{w}_k$ at the AP. Consider estimating each of the $D$ codeword elements in $\mathbf{w}$ individually, denoted by $w$, repeated for all elements in $\mathbf{w}$. The objective formally
    becomes
\begin{equation}\label{eq:weightsumestimateobjective}
	\begin{split}
		\underset{\widehat{w}}{\text{min}}\text{ }&\var\left(\widehat{w}|\textbf{W}\right),\\
		&\text{s.t. }\mathbb{E}\left[\widehat{w}|\textbf{W}\right] = w,\text{ }\mathbb{E}\left[|t_{k,l}|^2|w_k\right] = \rho_kw_k \leq P,
	\end{split}\tag{P2}
\end{equation}
where $\mathbf{W}=[\mathbf{w}_1\ldots\mathbf{w}_K]\in[w_\text{min},w_\text{max}]^{D\times K}$.

\subsection{Modulation for Non-Coherent Transmission and Reception}\label{sec:trans_and_rec}	
As stated in Section \ref{sec:ncota}, NC-OAC relies on amplitude
modulation of each transmitted symbol $t_{k,l}$. Similar to \cite{michelusi2024non, lee2024performance, hellstrom2023optimal, goldenbaum2013robust} we use the following modulation:
\begin{equation}
	t_{k,l} = \sqrt{w_k}s_{k,l},
\end{equation}
where $s_{k,l}=\sqrt{\rho_k}e^{j\phi_{k,l}}$ is the precoding
coefficient of device $k$ at time instant $l$. Here, $\phi_{k,l}\in\mathbb{R}$
is the transmit phase, and $\rho_k\in\mathbb{R}$ the power control coefficient of
device $k$. To ensure uncorrelated symbols over time
$l$, we consider $\phi_{k,l}\sim \text{U}(0,2\pi)$\footnote{The uniform distribution between $a>b\in\mathbb{R}$ is denoted by $\text{U}(b,a)$.}, i.i.d. across $k,l$. This implies that knowledge of $\phi_{k,l}$ is not necessary at the AP, which simplifies implementation significantly and enables using devices with unstable oscillators. Then at antenna $m$ and time instant $l$ the AP receives 
\begin{equation}\label{eq:rml}
	r_{m,l} = \sum_{k=1}^{K}g_{m,k}\sqrt{w_k}s_{k,l} + n_{m,l} \in \mathbb{C}.
\end{equation}
To compensate for a multiplicative bias at the AP, the received signal is amplified, yielding ${\widetilde{r}_{m,l} \equiv \sqrt{\eta} r_{m,l}}$ where $\eta\in\mathbb{R}_{\geq 0}$ is the receive scaling factor. The AP can then estimate $w$ using the codeword sum estimator in Definition \ref{def:codewordsumestimator}:
\begin{definition}[The Codeword Sum Estimator]\label{def:codewordsumestimator}
	The codeword sum estimator is defined as
	\begin{equation}\label{theestimator}
		\widehat{w} = \frac{\widetilde{\mathbf{r}}^H\widetilde{\mathbf{r}}}{ML} - \eta,
	\end{equation}
	where $\widetilde{\mathbf{r}}\in\mathbb{C}^{ML}$ is a vector consisting of all $\widetilde{r}_{m,l}$.
\end{definition}
As this estimator does not guarantee $\widehat{w}\in[Kw_\text{min},Kw_\text{max}]$, an alternative estimator (used indirectly in \cite{michelusi2024non}) is the projected estimator, defined as follows:
\begin{definition}[The Projected Codeword Sum Estimator]\label{def:projectedcodewordsumestimator}
	Let $\widehat{w}$ be as (\ref{theestimator}), then its projected version $\overline{\widehat{w}}$ is
	\begin{equation}\label{eq:projection}
		\overline{\widehat{w}} \equiv \begin{cases}
			Kw_\text{min} & \text{ if }\widehat{w}<Kw_\text{min},\\
			Kw_\text{max} & \text{ if } \widehat{w}>Kw_\text{max},\\
			\widehat{w} & \text{ else}.
		\end{cases}
	\end{equation}
\end{definition}
Compared to (\ref{theestimator}) the projected estimator is biased, and  more difficult to analyze.

Both estimators are motivated by the probability density function (PDF) of the received values $\widetilde{\mathbf{r}}$, particularly for the case of statistical CSI:
\subsubsection{Statistical CSI}\label{sec:statistical_CSI}
With statistical CSI $\beta_k$, we have
\begin{equation}\label{eq:SC_cross_correlation}
	\mathbb{E}\left[\widetilde{r}_{m',l'}^*\widetilde{r}_{m,l}|\mathbf{W}\right] = \begin{cases}
		\sigma^2 & (m',l')=(m,l),\\
		0 & \text{else},
	\end{cases}
\end{equation}
where 
\begin{equation}\label{eq:second_moment_gaussian_r}
	\sigma^2 = \eta\sum_{k=1}^{K} w_{k}\beta_k\rho_k + \eta,
\end{equation}
and
\begin{equation}
\widetilde{r}_{m,l}|\mathbf{W} \sim \mathcal{CN}\left(0, \sigma^2\right).
\end{equation}
Conditional on $\mathbf{W}$ each received symbol $\widetilde{r}_{m,l}$ is marginally Gaussian but only \textit{jointly} Gaussian across $m,l$ for $L=1$, as established by the following lemma. 
\begin{lemma}[Jointly Gaussian given $w$]\label{lemma:jointly_gaussian}
	 Under the channel model and modulation scheme given in Sections \ref{sec:channel_model}, \ref{sec:trans_and_rec} with $L=1$ and statistical CSI, the PDF of $\widetilde{\mathbf{r}}|\mathbf{W}$ is
		\begin{equation}\label{eq:MLPDF}
			p\left(\widetilde{\mathbf{r}}|\mathbf{W}\right) = \frac{1}{\pi^{M}(\sigma^2)^{M}}e^{-\frac{\widetilde{\mathbf{r}}^H\widetilde{\mathbf{r}}}{\sigma^2}}.
		\end{equation}
		Furthermore, if $\eta$ and $\rho_k$ are selected such that $\sigma^2-\eta=w$, then the estimator from Definition \ref{def:projectedcodewordsumestimator} is the maximum-likelihood estimator of $w$. Alternatively, if $w$ is relaxed to $\mathbb{R}$, then the maximum-likelihood estimator of $w$ corresponds to the estimator from Definition \ref{def:codewordsumestimator}.
\end{lemma}
\begin{proof}
	By definition, if any linear combination $\sum_{m}^{M}a_{m}\widetilde{r}_{m},\text{ } a_{m}\in\mathbb{C}$,
	is Gaussian, then $\widetilde{r}_{m}$ are jointly
    Gaussian across $m$, given $\mathbf{W}$
    \cite{edition2002probability}. This is the case here
    since $\widetilde{r}_{m}$ is a sum of independent Gaussian
    variables, and all $\widetilde{r}_{m}$ are mutually independent given $\mathbf{W}$ (note that this is not true for $L>1$). With cross-correlation
    (\ref{eq:SC_cross_correlation}), it follows
    that the joint PDF of the received symbols
    $\widetilde{r}_{m}$ is given by (\ref{eq:MLPDF}). The maximum-likelihood estimator follows from selecting $w$ as $\frac{\delta \text{log}(p\left(\widetilde{\mathbf{r}}|\mathbf{W}\right))}{\delta \overline{w}}=0$ if in the domain of $w$, otherwise as the maximizing boundary point, which yields (\ref{eq:projection}).
\end{proof}
\begin{remark}
In prior works listed by Table \ref{tab:mapping_applications}, the usage of $\widetilde{\mathbf{r}}^H\widetilde{\mathbf{r}}$ to estimate $w$ is generally motivated by the useful structure of its expected value $\sigma^2$ given by (\ref{eq:second_moment_gaussian_r}). Lemma \ref{lemma:jointly_gaussian} connects this heuristic second-moment-based estimator to a maximum-likelihood estimator. 
\end{remark}

\begin{remark}
	Finding $p\left(\widetilde{\mathbf{r}}|\mathbf{W}\right)$ for $L>1$ and random $\phi_{k,l}$ is difficult. However, by letting $\phi_{k,l}=0$, the maximum-likelihood estimator of $w$ can be found from the conditional PDF $p\left(\widetilde{\mathbf{r}}|\mathbf{W}, \{\phi_{k,l}\}\right)$, which is jointly Gaussian. The resulting biased maximum-likelihood estimator of $w$ is $\widehat{w}=\frac{1}{ML^2}\sum_{m=1}^{M}|\widetilde{\mathbf{r}}_m^\text{H}\mathbf{1}|^2-\frac{\eta}{L}$ with projection onto $[Kw_\text{min},Kw_\text{max}]$.
	Through simulations we have observed that this new estimator only has a tiny improvement over (\ref{theestimator}) in terms of $\mathbb{M}\mathbb{S}\mathbb{E}(\widehat{w})$ when $w_k\sim\text{U}[0,1]$ and $\phi_{k,l}=0$. However, when $\phi_{k,l}\sim\text{U}[0,2\pi]$,
	which is the case of practical interest, its performance is significantly worse than (\ref{theestimator}). Because of the lack of significant improvement, and interest of space, the simulations are not included herein.
\end{remark}

\subsubsection{Instantaneous CSI}
When $|g_{m,k}|$ is known, let $\mathbb{E}[.]$ and $\var(\cdot)$ represent the expectation and variance conditional 
on $|g_{m,k}|, \forall m, k$.\footnote{This reduces the notational burden in the paper. However, in the simulations we deviate from this notation by approximating the expectation over the instantaneous channels numerically.} We have 
\begin{equation}
	\begin{split}
		&\underbrace{\mathbb{E}\left[g_{m',k}^*g_{m,k}\right]}_{\text{Conditional on }|g_{m,k}|, \forall m, k}= \begin{cases}
			|g_{m,k}|^2 & m=m', \\
			0 & m\neq m',
		\end{cases}
	\end{split}
\end{equation}
which gives
\begin{equation}\label{eq:IC_crosscorr}
	\mathbb{E}\left[\widetilde{r}_{m',l'}^*\widetilde{r}_{m,l}|\mathbf{W}\right] = \begin{cases}
		\sigma_m^2 & (m',l')=(m,l),\\
		0 & \text{else},
	\end{cases}
\end{equation}
where
\begin{equation}\label{eq:IC_secondmoment}
	\sigma_m^2= \eta\sum_{k=1}^{K} w_{k}|g_{m,k}|^2\rho_k + \eta.
\end{equation}
Note that  $\widetilde{r}_{m,l}|\mathbf{W}$ conditional on the channel amplitude gain $|g_{m,k}|$ is \textit{not} Gaussian. Nevertheless, the codeword sum estimators (\ref{theestimator}) and (\ref{eq:projection}) will be applied in the instantaneous CSI case. 

\subsection{Power Control for Unbiased Estimation}\label{sec:estimate_w}
To solve (\ref{eq:weightsumestimateobjective}), Theorem \ref{theorem_cond_var_w} states that unbiased estimation is \textit{in practice} equivalent to ensuring that all devices are weighted equally at the AP, giving optimal power control $\rho_k$ and $\eta$:

\begin{theorem}[Conditional Variance of $\widehat{w}$]\label{theorem_cond_var_w}
	The codeword sum estimator $\widehat{w}$ defined in (\ref{theestimator}) with 
	\begin{equation}\label{eq:opt_parameters}
		\rho_k = \frac{1}{\eta\widehat{\beta_k}},\text{ }\eta = \frac{w_\text{max}}{P\widehat{\beta_\text{min}}},\text{ } \widehat{\beta_\text{min}}\equiv\underset{k}{\text{ min }}\widehat{\beta_k},
	\end{equation}
	where
	\begin{equation}\label{eq:betahattdef}
		\widehat{\beta_k} \equiv \begin{cases}
			\beta_k & \text{ with statistical CSI},\\
			\sum_{m=1}^M|g_{m,k}|^2/M & \text{ with instantaneous CSI},
		\end{cases}
	\end{equation}
	limited to $M=1$ under the instantaneous CSI scenario, minimizes objective (\ref{eq:weightsumestimateobjective}) with values listed by (\ref{eq:varwgivenw}). Furthermore, the transmitted codewords are weighted equally at the receiver: $\eta w_k\widehat{\beta_k}\rho_k = w_k$, which is sufficient and necessary to fulfill $\mathbb{E}[\widehat{w}|\mathbf{W}]=w$. The conditional variance is
	\begin{equation}\label{eq:varwgivenw}
		\var(\widehat{w}|\mathbf{W}) = 
		\begin{cases}
			\frac{1}{ML}\left(\sigma^4 + (L-1) \sum_{k=1}^Kw_k^2\right) & \text{w. SC},\\
			\frac{1}{L}\left(\sigma^4- \sum_{k=1}^Kw_k^2\right)& \text{w. IC},
		\end{cases}
	\end{equation}
	where $\sigma^4 = \left(w + \eta\right)^2$.
\end{theorem}

\begin{proof}Appendix \ref{proof_theorem_cond_var_w}.\end{proof}
From (\ref{eq:opt_parameters}) and (\ref{eq:varwgivenw}) one observes how performance is limited by the weakest channel, the maximum transmit power and number of repetitions over time. The variance also grows with $w$, reasonably caused by channel uncertainty in phase. For example, the error in phase can lead to a total elimination of the amplitude at the receiver, or a total alignment, meaning a variance proportional to $w$. Compared to instantaneous CSI, for statistical CSI the variance approaches the non-zero limit $\sum_{k=1}^Kw_k^2 $ as $L\rightarrow\infty$. An explanation is that power control based on statistical CSI has a misalignment error to the instantaneous channel amplitude, which ends up repeated $L$ times. For instantaneous CSI, the instantaneous power control decreases this misalignment error as $-w_k^2$. However, errors from random transmit phases persist, appearing as cross-terms of $w_k$. Finally, for instantaneous CSI,  (\ref{eq:varwgivenw}) is conditional on the channel amplitudes $|g_{m,k}|$, where the minimum value across $k$ is likely to be very small for large $K$, leading to a strong noise amplification.

\subsection{Codeword Sum Estimator With Projection}\label{sec:projected_estimator}
The projected estimator from Definition \ref{def:projectedcodewordsumestimator} is biased as it does not satisfy $\mathbb{E}\left[\overline{\widehat{w}}|\mathbf{W}\right]=w$. This bias depends on $w$ and is unavailable in closed-form even for $M=L=1$. Since the focus herein is on unbiased estimation, this projected estimator will not be investigated in depth. However, in the simplest case of $w=0,M=L=1$ it follows that $p(\eta|r_{1,1}|^2 - \eta >K)=e^{-\left(\frac{K}{\eta}+1\right)}$ giving a bias $\mathbb{E}\left[\overline{\widehat{w}}|\mathbf{W}\right]=Ke^{-\left(\frac{K}{\eta}+1\right)}\in[0,K/e]$. Here it is noteworthy that the bias is limited by $K/e$ (rather than $K$ as may be
intuitively expected). 

\subsection{Optimal Detection for Majority Voting and Counting}\label{sec:optimal_detection_MV_and_C}
The detector minimizing the probability of error for counting is the maximum a posteriori (MAP) detector, which detects the count giving the maximum posterior probability \cite{kay1998fundamentals}. For the Unit Affine mapping with $L=1$, the count posterior is $p(w=c|\mathbf{r})=p(\mathbf{r}|w=c)p(w=c)/p(\mathbf{r})$, where $p(\mathbf{r}|w=c)$ is given by Lemma \ref{lemma:jointly_gaussian} and $p(w=c)$ by the Binomial distribution. The normalizing factor $p(\mathbf{r})$ can be disregarded. Thus for this special case the optimal detector can be expressed in closed-form. Majority voting can also be optimally expressed as a MAP detector through $p(w\geq K/2|\mathbf{r})=p(w=\lceil{K/2}\rceil|\mathbf{r})+p(w=\lceil{K/2}\rceil+1|\mathbf{r})+\dots+p(w=K|\mathbf{r})$. However, showing a closed-form connection to (\ref{eq:sign_function_mv}) is not straightforward even for the Unit Affine mapping with $L=M=P=\beta_k=\eta=\rho_k=1,  K=3$ because of the summation of Gaussian PDFs.

\section{Estimating $x$ from $\hat{w}$}\label{sec:estimate_x}
Finally, given the estimated codeword sum $\widehat{\mathbf{w}}$, we estimate the source data sum $x$ through the decoder mappings introduced in Section \ref{sec:affine_mappings}. As the primary mappings of interest are the Affine and Augmented Affine mappings, which exist in prior work, these will be analyzed separately from the new Extended Affine mapping.

\begin{definition}[Sum Estimator]\label{def:sum_estimator}
	Given $\widehat{\mathbf{w}}$, the \textbf{sum estimator} $\widehat{x}$ is defined as follows:
	\begin{equation}
		\begin{split}
			&\widehat{x} = \mathcal{D}\left(\widehat{\mathbf{w}}\right).
		\end{split}
	\end{equation}
\end{definition}
Let $x_\text{min}=-1, x_\text{max}=1$, which can be shifted and scaled arbitrarily. Furthermore, let $x_+ = \sum_{k=1}^K(x_k)^+$, ${x_-=\sum_{k=1}^{K}(-x_k)^+}$, such that $x=x_+-x_-$. Now, consider the Unit Affine and Augmented Affine mappings from Examples \ref{simple_A} and \ref{simple_AA}. The system has to aggregate one codeword $w$ for the Affine mapping, and two codewords $w_1, w_2$ for the Augmented Affine:
\begin{equation}\label{simplemappings}
	\begin{split}
		&\text{A: }w = \frac{x+K}{2},\text{ AA: }w_1 = x_+,\text{ }w_2 = x_-.
	\end{split}
\end{equation}
By Definition \ref{def:sum_estimator}, the estimators of $x$ are given by:
\begin{equation}\label{xestimator}
	\begin{split}
		&\text{A: }\widehat{x} = 2\widehat{w} - K,\text{ AA: }\widehat{x} = \widehat{w_1} - \widehat{w_2},
	\end{split}
\end{equation}
where $\widehat{w}, \widehat{w_1}, \widehat{w_2}$ are the raw estimates of $w,w_1,w_2$ from (\ref{theestimator}) without the projection in (\ref{eq:projection}) . 

\begin{remark}
	If $\mathbb{E}[\widehat{\mathbf{w}}|\mathbf{W}]=\mathbf{w}$, then (\ref{simplemappings}) and (\ref{xestimator}) give ${\mathbb{E}[\widehat{x}|\mathbf{x}]=x}$, as shown in the proof of Theorem \ref{theorem_total_variance}. 
\end{remark}

\subsection{Conditional Variance of $\widehat{x}$}
To compare $\var(\widehat{x}|\mathbf{x})$ for the Affine and Augmented Affine mapping given by Theorem \ref{theorem_cond_variance}, the Augmented Affine mapping is allocated half the resources $(L/2)$ (i.e., channel uses) per codeword since the number of codewords is twice that of the Affine mapping ($D=2$ compared to $D=1$). 

\begin{theorem}[Conditional Variance of $\widehat{x}$]\label{theorem_cond_variance}
	With $\widehat{w}$ from (\ref{theestimator}), $\rho_k, \eta$ from Theorem \ref{theorem_cond_var_w}, $x_k\in[-1,1], \forall k$, the two affine mappings described in (\ref{simplemappings}), (\ref{xestimator}), the conditional variance of $\widehat{x}$ is given by the expressions in Table \ref{tab:varxgivenx}. The Augmented Affine mapping is given $L/2$ resources per codeword, the Affine is given $L$.
	\begin{table}[!h]
		\begin{center}
			\renewcommand{\arraystretch}{1.5} 
			\setlength{\tabcolsep}{7pt}      
			\caption{Conditional variance of $\widehat{x}$.}
			\label{tab:varxgivenx}
			\begin{tabular}{ | c |c |c  | } 
				\hline
				CSI& Map & $\var(\widehat{x}|\mathbf{x})$  \\ 
				\hline
				SC & A & $\frac{1}{ML}\left(\left(x + K +2\eta\right)^2 + (L-1)\sum_{k=1}^K\left(x_k+1\right)^2\right)$ \\ 
				\hline
				SC & AA & $\frac{2}{ML}\left(\left(x_++\eta\right)^2 + \left(x_- + \eta\right)^2+\frac{L-2}{2}\sum_{k=1}^{K}x_{k}^2\right)$ \\ 
				\hline
				IC & A & $\frac{1}{L}\left(\left(x+K + 2\eta\right)^2 - \sum_{k=1}^K\left(x_k+1\right)^2\right)$ \\ 
				\hline
				IC & AA & $\frac{2}{L}\left(\left(x_+ + \eta\right)^2 + \left(x_- + \eta\right)^2 - \sum_{k=1}^Kx_k^2\right)$ \\ 
				\hline
			\end{tabular}
		\end{center}
	\end{table}
\end{theorem}
\begin{proof}Appendix \ref{proof_theorem_cond_variance}.\end{proof}
In Theorem \ref{theorem_cond_variance} the $\var(\widehat{x}|\mathbf{x})$ of the Affine mapping has an explicit quadratic growth in the number of devices $K$. For the Augmented Affine mapping this quadratic growth is implicit through $x$, but appears explicitly in Theorem \ref{theorem_total_variance}. This implies a useful growth in received power proportional to $K$ since measurement noise is not proportional to $K$. The four $\var(\widehat{x}|\mathbf{x})$ are plotted in Fig. \ref{fig:A_vs_AA} and discussed further in Section \ref{sec:sim_cond_variance}.
\subsection{Variance of $\widehat{x}$}
Theorem \ref{theorem_total_variance} gives the total variance of $\widehat{x}$ when $x_k\sim\text{U}(-1,1)$, i.i.d. across $k$ following from the \textit{law of total variance}:
\begin{equation}\label{eq:total_variance}
	\var(\widehat{x}) = \mathbb{E}\left[\var(\widehat{x}|\mathbf{x})\right] + \var(\mathbb{E}[\widehat{x}|\mathbf{x}]).
\end{equation}

To ensure a fair comparison between the total variance of the Affine and Augmented Affine mapping, $x_k$ must be equally likely in either half of the Augmented Affine split, which is the case for $\text{U}(-1,1)$. Furthermore, both mappings should use the same expected power, which according to Proposition \ref{prop:fair_transmission} holds for the Augmented Affine mapping ($N=2$) when $\rho_{k,\text{AA}}=2\rho_k, \eta_\text{AA}=\eta/2$.  

\begin{theorem}[Variance of $\widehat{x}$ under Uniform Data]\label{theorem_total_variance}
	Let $x_k\sim\text{U}(-1,1)$, i.i.d. across $k$ where $\eta_\text{AA}=\eta/2$. Then the total variance of the conditional variances presented in Theorem \ref{theorem_cond_variance} is given by the expressions in Table \ref{tab:vartotx}. Additionally, the Unit Affine mapping (\ref{eq:simple_affine_mapping_coeff}) achieves the minimum $\var(\widehat{x})$ for any valid choice of $a,b,c,d$ in the Affine mapping.
	\begin{table}[!h]
		\begin{center}
			\renewcommand{\arraystretch}{1.5} 
			\setlength{\tabcolsep}{10pt}      
			
			\caption{Total variance of $\widehat{x}$ with uniformly distributed $x_k$.}
			\label{tab:vartotx}
			\begin{tabular}{ | c |c |c  | } 
				\hline
				CSI& Map & $\var(\widehat{x})-\frac{K}{3}=\mathbb{M}\mathbb{S}\mathbb{E}(\widehat{x})$  \\ 
				\hline
				SC & A & $\frac{1}{ML}\left(K^2 -K + 4K\eta + 4\eta^2 + \frac{4KL}{3}\right)$ \\ 
				\hline
				SC & AA & $\frac{1}{ML}\left(\frac{K^2}{4} - \frac{K}{4} + 2K\eta_\text{AA} + 4\eta_\text{AA}^2+\frac{KL}{3}\right)$ \\ 
				\hline
				IC & A & $\frac{1}{L}\left(K^2 - K + 4K\eta + 4\eta^2\right)$ \\ 
				\hline
				IC & AA & $\frac{1}{L}\left(\frac{K^2}{4} - \frac{K}{4}+ 2K\eta_\text{AA} + 4\eta_\text{AA}^2\right)$ \\ 
				\hline
			\end{tabular}
		\end{center}
	\end{table}
\end{theorem}
Note that Table \ref{tab:vartotx} contains $\mathbb{M}\mathbb{S}\mathbb{E}(\widehat{x})$, which is equal to $\var(\widehat{x})$ given an unbiased estimator as stated by Remark \ref{rem:MSE} where $\var(x)=K/3$. The main takeaway from Theorem \ref{theorem_total_variance} is summarized in the following corollary.
\begin{corollary}[Augmented Affine Beats Affine under Uniform Data]\label{corollary_AA_better_than_A}
	For all $\eta>0$, $K\geq 1$ and $L\geq 1$
	\begin{equation*}
		\begin{split}
			&\var(\widehat{x})_{\text{SC,A}} > \var(\widehat{x})_{\text{SC,AA}},\text{ }\var(\widehat{x})_{\text{IC,A}} > \var(\widehat{x})_{\text{IC,AA}},
		\end{split}
	\end{equation*}
	where $\var(\widehat{x})_\text{SC,A}$ means statistical CSI knowledge and the Affine mapping, for example. 
\end{corollary}
\begin{proof}
	Follows directly from Theorem \ref{theorem_total_variance}.
\end{proof}

While the coefficient choices in (\ref{eq:simple_augmented_affine_coeff}) are not necessarily optimal for the Augmented Affine mapping, its $\var(\widehat{x})$ is still lower than the variance achieved by the Unit Affine mapping shown by Theorem \ref{theorem_total_variance} to be optimal. Thus, Corollary \ref{corollary_AA_better_than_A} shows that the Augmented Affine is better than the Affine mapping for uniform data. However, as can be shown through Theorem \ref{theorem:generalized_affine_mapping} for the Extended Affine mapping, the Unit Augmented Affine ($N=2$) is not optimal for all $L$:

\begin{theorem}[Variance of the Extended Affine Mapping]\label{theorem:generalized_affine_mapping}
	Let $x_k\sim\text{U}(-1,1)$, i.i.d. across $k$, then the total variance of $\widehat{x}$ using the Extended Affine mapping with encoder (\ref{eq:AA_N_encoder}), decoder (\ref{eq:AA_N_decoder}) and codewords aggregated using the codeword sum estimator (\ref{theestimator}) with statistical CSI is
	\begin{equation*}
		\begin{split}
			&\var(\widehat{x})-\frac{K}{3} = \mathbb{M}\mathbb{S}\mathbb{E}(\widehat{x}) =\\
			&=\frac{4}{N^2ML_w}\bigg(\frac{K(K-1)}{4N} + K\eta_N + N\eta_N^2 + \frac{KL_w}{3}\bigg)+\\
			&\left(\frac{N^2-3N+2}{3NML_b}\right)\left(\frac{K(K-1)}{N^2} + \frac{K(2\eta_N +L_b)}{N} + \eta_N^2\right),
		\end{split}
	\end{equation*}
	where $L=NL_w + (N-2)L_b$. $L_w$ and $L_b$ are positive integers representing the number of repetitions per continuous codeword and indicator codeword, respectively, where $L_b=0$ if $N=2$. Furthermore, $\eta_N$ is from Proposition \ref{prop:fair_transmission}.
\end{theorem}
\begin{proof}
	Appendix \ref{proof_generalized_affine_mapping}.
\end{proof}
Finding an optimal $N$ for general $L$ in closed form is not straightforward as it requires solving the Diophantine equation $L= NL_w + (N-2)L_b$. Even if multiple feasible $N$ are found given $L$, its important to ensure the mappings with different $N$ expend the same amount of power, as given by Proposition \ref{prop:fair_transmission}.
\begin{proposition}[Transmission Power Normalization]\label{prop:fair_transmission}
	Let $x_k\sim\text{U}(-1,1)$, $t_{k,l}$ one of the $L$ symbols transmitted in (\ref{eq:channel_model}). Then for normalization of total energy with respect to the Unit Affine mapping with power control $\rho_k, \eta$ from Theorem \ref{theorem_cond_var_w}, the Extended Affine mapping should use $\rho_{k,N}$ and $\eta_{N}$ as
	\begin{equation*}
		\begin{split}
			&\rho_{k,N} = \rho_k\frac{NL}{L + L_b(N-2)},\\
			&\Leftrightarrow\\
			&\eta_{N}=\eta \frac{L + L_b(N-2)}{NL}<\eta, \forall N,
		\end{split}
	\end{equation*}
	where $L_b=0$ for $N=2$.
\end{proposition}
\begin{proof}
	Appendix \ref{proof_fair_transmission}.
\end{proof}
Interestingly, Proposition \ref{prop:fair_transmission} implies that without normalization the energy consumption of the Extended Affine is always lower than the Unit Affine mapping. 
Numerical evidence from simulations with power normalization in Section \ref{sec:ext_affine_sim} show that either $N=2$ or $N=4$ is optimal depending on $L$. Analytically, for asymptotically large $L$, and thus large $L_w, L_b$, the minimizing $N$ of the variance in Theorem \ref{theorem:generalized_affine_mapping} with power normalization is indeed $N=4$. For general $L$ this remains a conjecture. Finally, note that the performance of the Extended Affine mapping can be improved further by using a count detector as described in Section \ref{sec:majority_voting}, rather than the relaxed sum estimator (\ref{theestimator}), which is left to future work.

\section{Numerical Evaluation}\label{simulation}
In all numerical  examples, power is set to $P=1$ and the large-scale fading coefficient $\beta_k\equiv\beta$ is equal for all devices. Then $\eta = w_\text{max}/\beta$ as given by (\ref{eq:opt_parameters}) which can be interpreted as the signal-to-noise-ratio (SNR). Note that this is representative of the case with heterogeneous channel conditions since all devices adapt to the weakest channel by design. The Affine mapping is given $L$ symbols per codeword, but the Augmented Affine $L/2$. For the instantaneous CSI case the simulations are \textit{not} conditional on $|g_{m,k}|$, rather the total variance is evaluated. Finally, for cases where the source data sum estimator is biased, such as the projected estimator, we present its MSE instead of its variance. 

\subsection{Estimator Variance Given Equal Data Across Devices}\label{sec:sim_cond_variance}
Fig. \ref{fig:A_vs_AA} shows a numerical evaluation of $\var(\widehat{x}|x)$ from Theorem \ref{theorem_cond_variance} for different values of $x_k = x/K$ such that all $x_k$ are equal.
\begin{figure}[t!]
	\centering
	\includegraphics[width=0.85\linewidth]{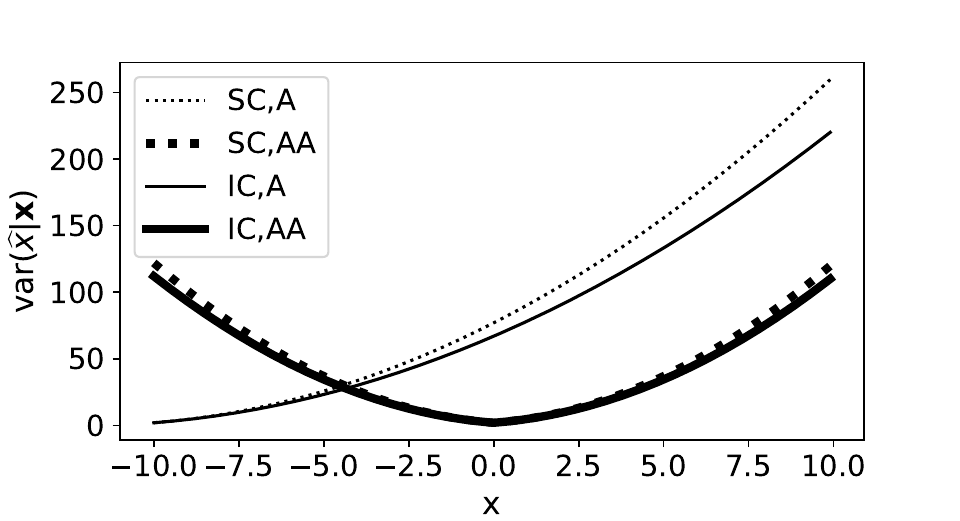}
	\caption{$\var(\widehat{x}|x_k=x/K)$ under the Affine vs Augmented Affine mapping against $x$. Computed according to Theorem \ref{theorem_cond_variance}. ${K=10}, {L=2},{M=1}, {\eta=1}$ (without energy normalization).}\label{fig:A_vs_AA}
\end{figure}
Here no energy normalization is applied since it is conditional on $x$. It shows that none of the two mappings has exclusively lower variance conditional on $x$. Instead, the choice between the mappings depends on the distribution of $x$ in the range $[-10, 10]$: If $x$ tends to the lower values the Affine mapping is better, otherwise the Augmented Affine is better. Because of this ambiguity, looking at the total variance $\var(\widehat{x})$ as done in Fig. \ref{fig:MSE_A_vs_AA} will be more useful. The intuition is that there is a trade-off between uncertainty from phase and uncertainty from additive noise. The first is reduced by using multiple codewords as in the Augmented Affine mapping such that fewer devices transmit on each codeword, increasing the probability of alignment in phase. The second is reduced by repeating the codewords as in the Affine mapping such that noise is increasingly averaged to zero. For instantaneous CSI the improvement in variance over statistical CSI is marginal and $|g_\text{min}|$ is likely small making the variance high, suggesting that instantaneous CSI is unnecessary for practical applications. This tradeoff between phase uncertainty and additive noise, along with Fig. \ref{fig:A_vs_AA} was the inspiration for the Extended Affine mapping. Fig. \ref{fig:A_vs_AA} suggests that more segments lead to an increasingly flat $\var(\widehat{x}|\mathbf{x})$, but not necessarily lower, when the number of splits increases. 

\subsection{Estimator Total Variance With Uniformly Distributed Data}
Fig. \ref{fig:MSE_A_vs_AA} shows a numerical evaluation of $\mathbb{M}\mathbb{S}\mathbb{E}(\widehat{x})$ for different values of $1/\eta$ based on Theorem \ref{theorem_total_variance}.
\begin{figure}[t!]
	\centering
	\includegraphics[width=0.85\linewidth]{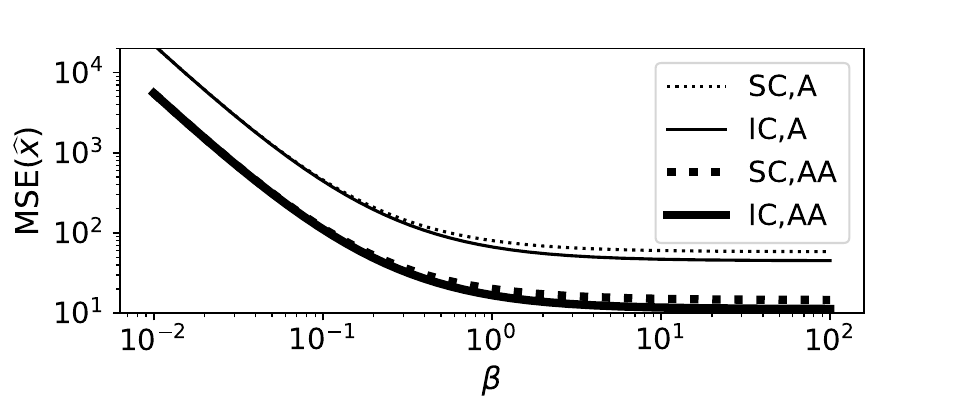}
	\caption{$\mathbb{M}\mathbb{S}\mathbb{E}(\widehat{x}) = \var(\widehat{x})-K/3$ of the Affine vs Augmented Affine mapping against ${\beta=1/\eta}$. Computed according to Theorem \ref{theorem_total_variance}. ${K=10},{M=1},L=2$.}\label{fig:MSE_A_vs_AA}
\end{figure}
As   stated by Corollary \ref{corollary_AA_better_than_A}, Fig. \ref{fig:MSE_A_vs_AA} shows how the Augmented Affine mapping has a lower estimation variance than the Affine for both instantaneous and statistical CSI under the uniform distribution of $x_k$. Fig. \ref{fig:MSE_A_vs_AA} also follows the corresponding numerically approximated MSE in Fig. \ref{fig:MSE_A_vs_AA_proj}.

\subsection{Estimator Variance With Statistical CSI and Projection}
\begin{figure}[t!]
	\centering
	\includegraphics[width=0.85\linewidth]{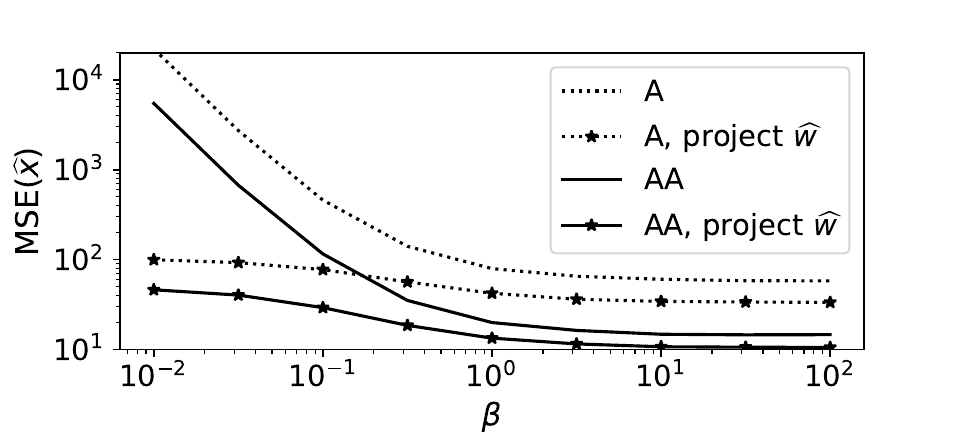}
	\caption{$\mathbb{M}\mathbb{S}\mathbb{E}(\widehat{x})$ of the Affine vs Augmented Affine mapping against ${\beta=1/\eta}$ with and without projection of $\widehat{w}$ onto $[0,Kw_\text{max}]$. Statistical CSI used. Number of trials $=10^5$, ${K=10, L=2, M=1}, x_k\sim\text{U}(-1, 1)$. \textit{Not} conditional on $|g_{m,k}|$.}\label{fig:MSE_A_vs_AA_proj}
\end{figure}
\begin{figure}[t!]
\centering
\includegraphics[width=0.85\linewidth]{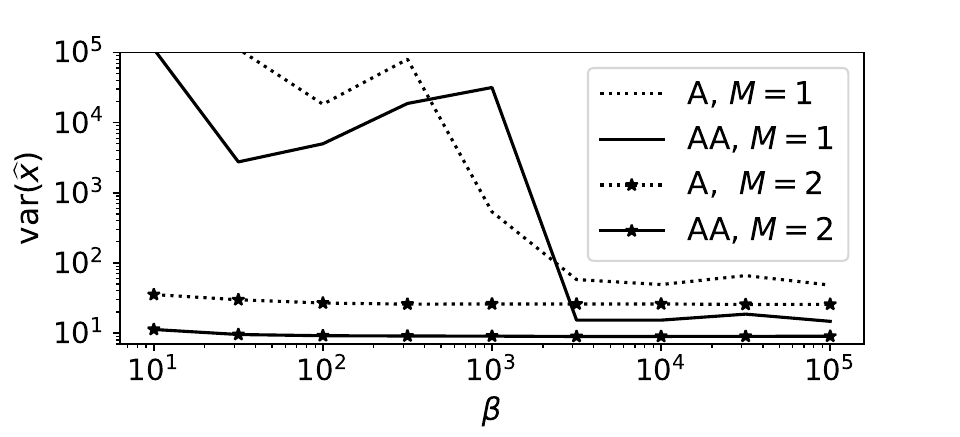}
\caption{$\var(\widehat{x})$ of the Affine vs Augmented Affine mapping against $1/\eta=\beta$. Instantaneous CSI used. For $M=1$, the simulated variance is theoretically infinite, but demonstrates its erratic behavior. Number of trials $=10^5$, ${K=10},L=2, x_k\sim\text{U}(-1, 1)$. \textit{Not} conditional on $|g_{m,k}|$.}\label{fig:MSE_IC}
\end{figure}

\begin{figure}[t!]
	\centering
	\includegraphics[width=0.85\linewidth]{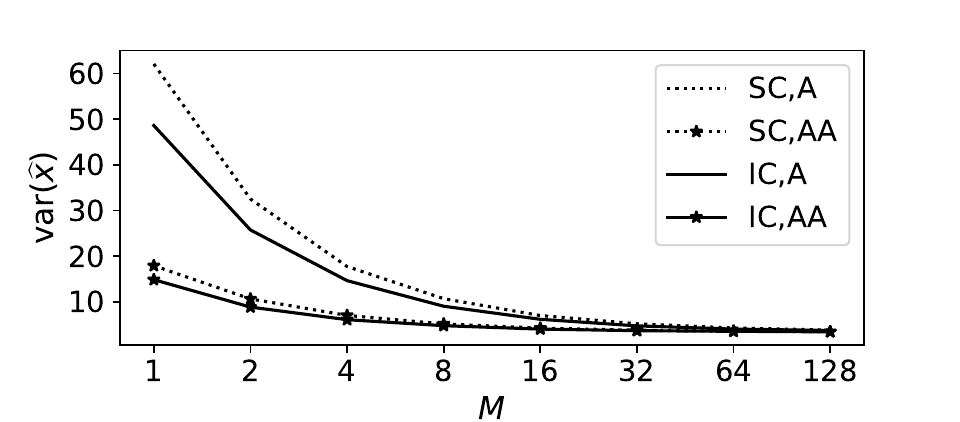}
	\caption{$\var(\widehat{x})$  of the mappings without projection under statistical and instantaneous CSI against the number of receive antennas $M$. Number of trials $=10^5$, ${K=10, L=2, \beta=10^4}, x_k\sim\text{U}(-1, 1)$. \textit{Not} conditional on $|g_{m,k}|$.}\label{fig:MSE_var_M}
\end{figure}
\begin{figure}[t!]
\centering
\includegraphics[width=0.85\linewidth]{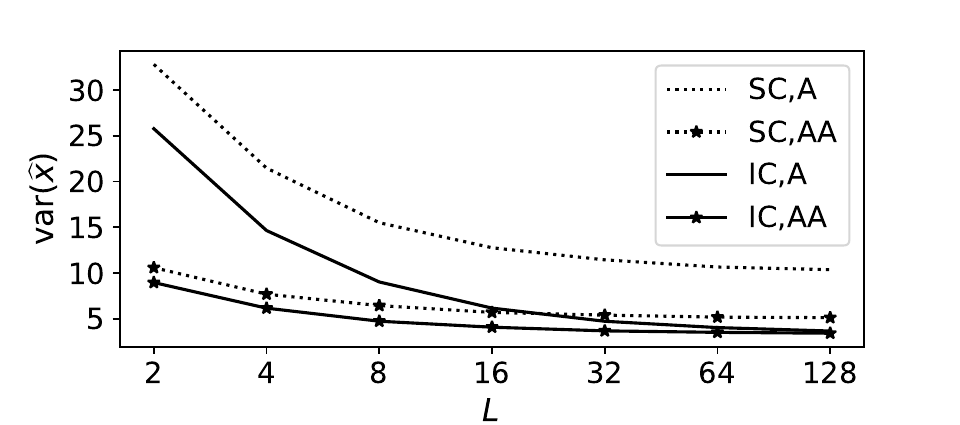}
\caption{$\var(\widehat{x})$  of the mappings without projection under statistical and instantaneous CSI against the number transmissions $L$. Number of trials $=10^5$, ${K=10, M=2, \beta=10^4}, x_k\sim\text{U}(-1, 1)$. \textit{Not} conditional on $|g_{m,k}|$.}\label{fig:MSE_var_L}
\end{figure}
\begin{figure}[t!]
\centering
\includegraphics[width=0.85\linewidth]{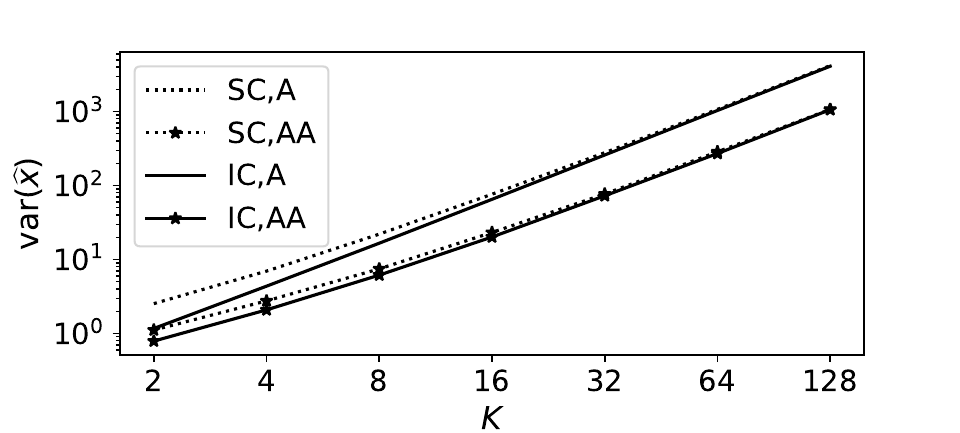}
\caption{$\var(\widehat{x})$  of the mappings without projection under statistical and instantaneous CSI against the number of devices $K$. Number of trials $=10^5$, $M=2, L=2, \beta=10^4, x_k\sim\text{U}(-1, 1)$. \textit{Not} conditional on $|g_{m,k}|$.}\label{fig:MSE_var_K}
\end{figure}
Fig. \ref{fig:MSE_A_vs_AA_proj} presents a simulation of $\mathbb{M}\mathbb{S}\mathbb{E}(\widehat{x})$ as a function of $\beta$ where the projected estimator (\ref{eq:projection}) is used.  Fig. \ref{fig:MSE_A_vs_AA_proj} shows that the projected estimator $\overline{\widehat{w}}\in[0, Kw_\text{max}]$ significantly reduces $\mathbb{M}\mathbb{S}\mathbb{E}(\widehat{x})$, particularly at low SNR $\beta$. However, the $\mathbb{M}\mathbb{S}\mathbb{E}(\widehat{x})$ without projection does not converge to $\mathbb{M}\mathbb{S}\mathbb{E}(\widehat{x})$ with projection as the SNR grows since the uncertainty in the amplitude gain $|g_{m,k}|$ relative to $\beta$ remains as $\beta$ grows. 
Furthermore, the Augmented Affine mapping is again better than the Affine, as stated by Corollary \ref{corollary_AA_better_than_A}. Finally, Fig. \ref{fig:MSE_A_vs_AA_proj} shows the simulated result which closely resembles the theoretical $\mathbb{M}\mathbb{S}\mathbb{E}(\widehat{x})$ from Theorem \ref{theorem_total_variance}, shown in Fig. \ref{fig:MSE_A_vs_AA}.

\subsection{Estimator Variance With Instantaneous CSI}
Fig. \ref{fig:MSE_IC} presents the simulated total variance for different $\beta$ under instantaneous CSI. The erratic behavior of the variance when $M=1$ is caused by the division of $|g_\text{min}|$  in the estimator (\ref{theestimator}) by $\widetilde{\mathbf{r}}$, which does not have a finite first moment. This curve is included only to demonstrate this fact. However, with $M=2$ the estimator performs well. While not included in Fig. \ref{fig:MSE_IC}, the projected estimator from (\ref{eq:projection}) was simulated with instantaneous CSI, which decreased the MSE similar to the projected estimator with statistical CSI in Fig. \ref{fig:MSE_A_vs_AA_proj}.

\subsection{Simulation of Estimator Variance for Various M, L, K}
Figures \ref{fig:MSE_var_M}, \ref{fig:MSE_var_L}, and \ref{fig:MSE_var_K} show the variance of the proposed estimators for various $M, L, K$.  When $M$ grows the variance of the estimator with instantaneous CSI will approach that of statistical CSI since then $\widehat{\beta_k}\approx \beta$, as implicitly given by (\ref{eq:betahattdef}). Increasing $L$ benefits the case with instantaneous CSI more than statistical CSI which has a non-zero limit for high $L$, which is in accordance with Theorem \ref{theorem_total_variance}. The  variance grows approximately linearly with  $K$ on a log-log scale, in accordance with Theorem \ref{theorem_total_variance}. 

\subsection{Extended Affine Mapping Versus Augmented Affine}\label{sec:ext_affine_sim}
Fig. \ref{fig:extended_vs_AA} presents the performance of the Extended Affine mapping, computed according to Theorem \ref{theorem:generalized_affine_mapping} where $N, L_w, L_b$ have been optimized numerically for each $L$ by grid-search. Interestingly, $N=2$ is optimal up to a certain $L$, where $N=4$ becomes optimal, this can be seen as the curve corresponding to the Extended Affine deviates from the Augmented Affine curve. Power normalization according to Proposition \ref{prop:fair_transmission} is applied, however, even without normalization $N=4$ remains optimal for sufficiently large $L$.
\begin{figure}[t!]
	\centering
	\includegraphics[width=0.85\linewidth]{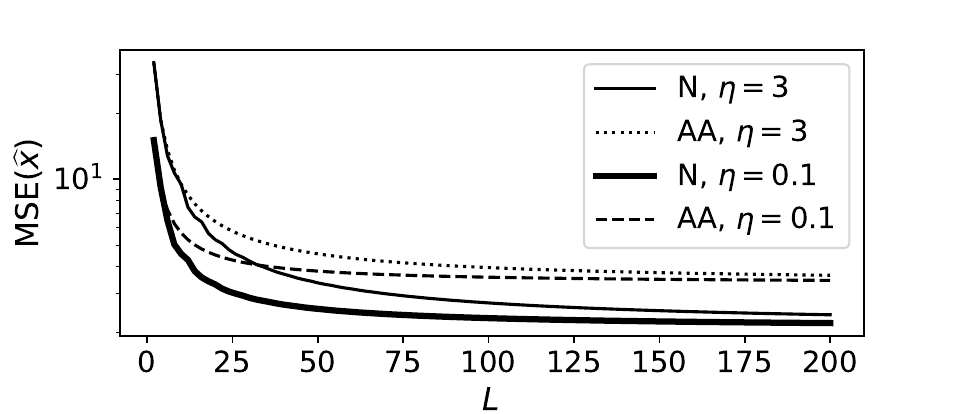}
	\caption{$\mathbb{M}\mathbb{S}\mathbb{E}(\widehat{x}) = \var(\widehat{x})-K/3$ of the Extended Affine (N) vs Augmented Affine mapping ($N=2$) against $L$, where the extended has the optimal $N, L_b, L_w$. Computed according to Theorem \ref{theorem:generalized_affine_mapping}. ${K=10}, {M=1}, {x_k\sim\text{U}(-1, 1)}$.}\label{fig:extended_vs_AA}
\end{figure}

\subsection{Impact of Source Data Distribution}
To simulate heterogeneous source data across devices, let $x_k\sim\mathcal{N}(b_k,1)$ where $K=10$ and $b_k$ is uniformly spaced in $[-2,2]$: $b_1=-2, b_2=-2+4/9, \dots b_{10}=2$. Before transmission $x_k$ is projected onto $[x_\text{min}, x_\text{max}]=[-1,1]$, this leads to a bias $(\mathbb{E}[\widehat{x}|\mathbf{x}]\neq x)$; hence,   $\mathbb{M}\mathbb{S}\mathbb{E}(\widehat{x})$ instead of $\var(\widehat{x})$ is simulated. Given by Fig. \ref{fig:non_iid_source_data}, the Unit Augmented Affine mapping still achieves lower estimation error than the Unit Affine mapping. An intuitive reason is that the average $b_k$ is near 0 which gives a symmetric distribution of $x$ around $0$, which by Fig. \ref{fig:A_vs_AA} is suited to the Augmented Affine mapping.

For simulation of extreme data distributions, which may occur in real applications such as FL, consider the following i.i.d. source data: Cauchy (CY) where $x_k$ are standard Cauchy distributed, Log-Normal (LN) where $x_k=e^{\phi_k}, \phi_k\sim\mathcal{CN}(0,1)$, Shifted Uniform (SU) where $x_k\sim\text{U}(-2,0)$ and finally a Binomial (B) with 10 draws normalized as $x_k=(q_k-5)/5, q_k\sim\text{B}(n=10, p=0.5)\in\{0,1,\dots,10\}$. 
Let $[x_\text{min}, x_\text{max}]=[-1,1]$ with projection before transmission as with heterogeneous source data. The Cauchy and Log-Normal cases represent data with heavy-tailed distributions, where CY can be both negative and positive, LN only positive. Shifted Uniform represents the case where all data is in the low end. Binomial represents a setting with sparsely distributed data.  Based on Fig. \ref{fig:weird_source_data} the used mappings do not handle the extreme outliers of the Cauchy source data in terms of MSE (A and AA overlapping in the figure). However, this is expected as the Cauchy distribution does not have a finite mean or variance. Nevertheless, the poor performance with heavy-tailed data suggests the need of a robust mapping for NC-OAC. For LN the performance is also poor, but the relation between A and AA is consistent with previous results. For SU the relation between A and AA is flipped as predicted by Fig. \ref{fig:A_vs_AA}, while the relation still holds for the normalized Binomial case. 

\begin{figure}[t!]
	\centering
	\includegraphics[width=0.85\linewidth]{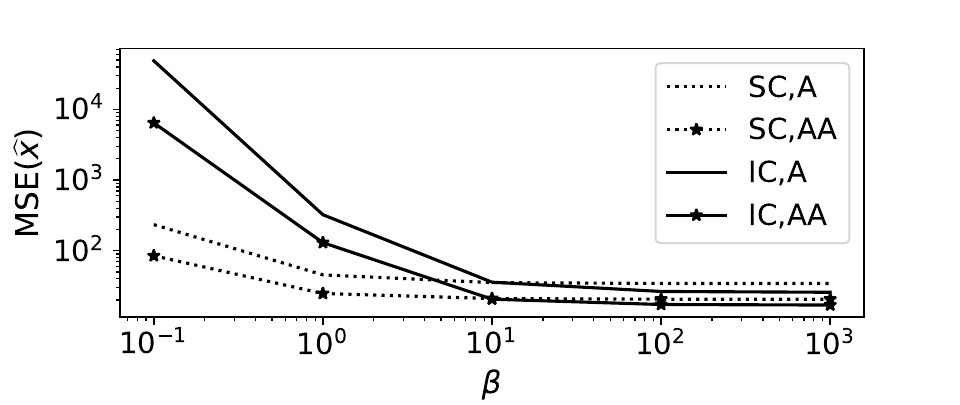}
	\caption{$\mathbb{M}\mathbb{S}\mathbb{E}(\widehat{x})$ against $\beta$ with $x_k\sim\mathcal{CN}(b_k, 1)$ where ${b_k}$ uniformly spaced in $[-2,2]$ including edges. $M=2,L=2, x_\text{min}=-1, x_\text{max}=1, K=10$. Number of trials $=10^6$.}\label{fig:non_iid_source_data}
\end{figure}	

\begin{figure}[t!]
\centering
\includegraphics[width=0.85\linewidth]{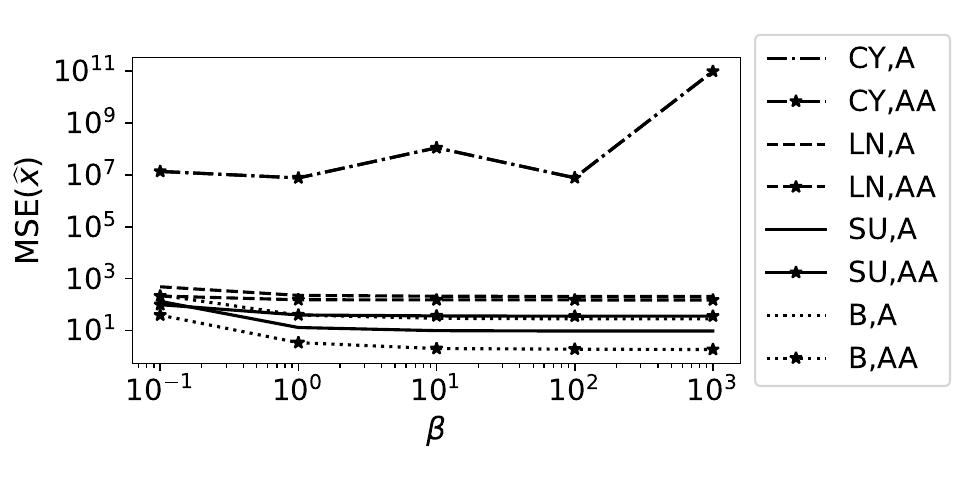}
\caption{$\mathbb{M}\mathbb{S}\mathbb{E}(\widehat{x})$ against $\beta$ for various extreme source data distributions. $M=2,L=2, x_\text{min}=-1, x_\text{max}=1, K=10$. Number of trials $=10^5$.}\label{fig:weird_source_data}
\end{figure}

\subsection{Majority Voting}
For majority voting, consider $x_k\in\{-1,1\}$ with $\text{P}(x_k=1)=p$, $\text{P}(x_k=-1)=1-p$ with vote probability $p$. In Fig. \ref{fig:MV_accuracy_p} the approximated vote accuracy against $p$ is presented. First, similar to Fig. \ref{fig:A_vs_AA} the best mapping for majority voting (Affine or Augmented Affine) depends on if the expected vote is more positive ($1$) or negative ($-1$) depending on $p$. Considering the minimum and average accuracy across $p$ in Fig. \ref{fig:MV_min_vs_avg}, the Augmented Affine is again superior. Furthermore, both the average and minimum majority vote accuracy approaches 1 as the number of antennas $M$ grows. Note that if $K$ was even the accuracy would not approach $1$ since the mappings do not handle a tie vote when noise is present; in this case the outcome would be random. No simulation of counting is included, as previously mentioned sum estimation is an approximate count.

\begin{figure}[t!]
	\centering
	\includegraphics[width=0.85\linewidth]{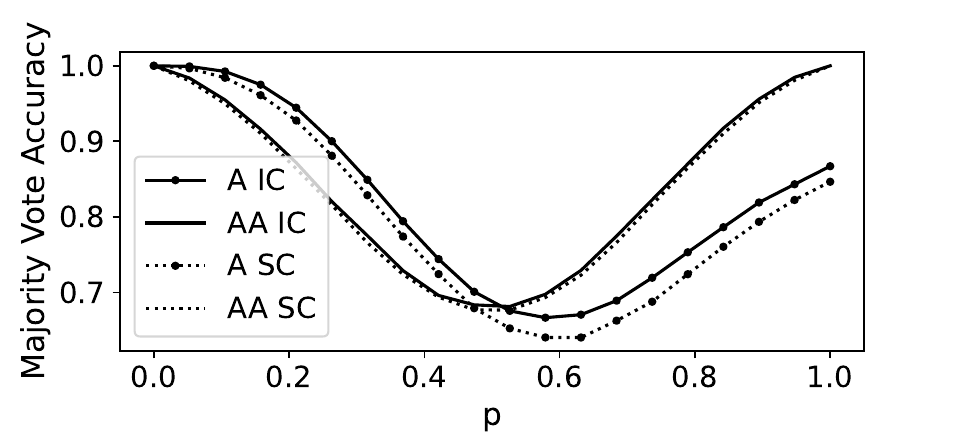}
	\caption{Majority vote accuracy against $p\in[0,1]$. Number of trials $=10^5$. $K=11, M=2, L=2, \beta=10^3$.}\label{fig:MV_accuracy_p}
\end{figure}

\begin{figure}[t!]
	\centering
	\includegraphics[width=0.85\linewidth]{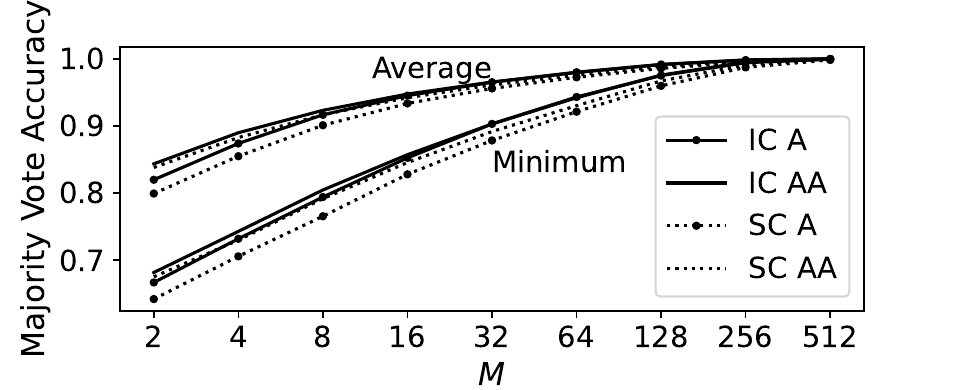}
	\caption{Minimum and average majority vote accuracy across 20 uniformly spaced vote probabilities $p\in[0,1]$ against $M$. Number of trials $=10^5$. $K=11, L=2, \beta=10^3$.}\label{fig:MV_min_vs_avg}
\end{figure}

\subsection{Application to FL With Projected Gradients}
We apply the NC-OAC framework with the unit mappings (\ref{eq:simple_affine_mapping_coeff}) and (\ref{eq:simple_augmented_affine_coeff}) for aggregating the local gradients in FL. The learning task is to classify digits from the MNIST dataset \cite{lecun2002gradient}. The learning model is a multilayer-perceptron with input size $28\times28$, hidden layer size $10$ with rectified linear unit activation, output size $10$ with softmax activation and cross-entropy loss \cite{Goodfellow-et-al-2016}. Each device holds $5000$ training samples from a single digit class. 
The NC-OAC projects the gradients onto $[-2,2]$ before transmission and the receiver exploits the projected sum estimator using (\ref{eq:projection}). The \textit{genie} computes $\nabla F\left(\boldsymbol{\theta}^t\right)$ exactly for each data batch. The results in Fig. \ref{fig:FL} show that the superiority of the Augmented Affine holds, which is consistent with Corollary \ref{corollary_AA_better_than_A}. Furthermore, statistical CSI performs not much worse than instantaneous CSI, or even the genie case. Details on FL can be found in Appendix \ref{sec:FL_mse}.
\begin{figure}[t!]
	\centering
	\includegraphics[width=0.85\linewidth]{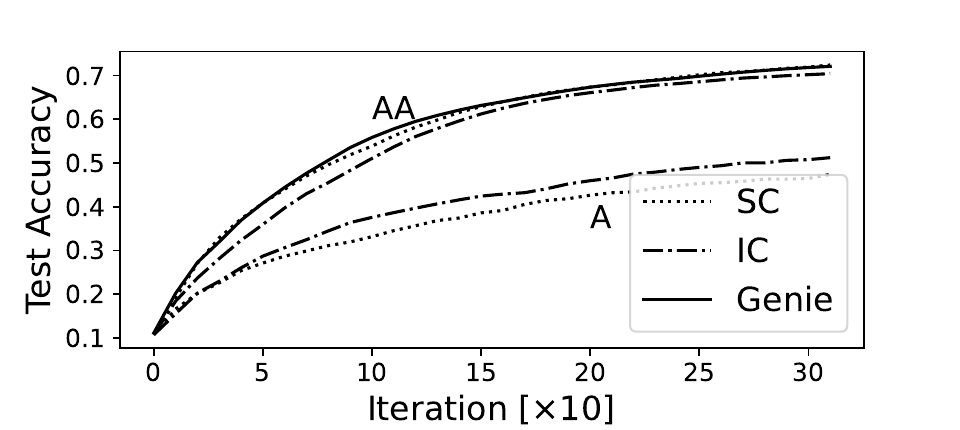}
	\caption{Test Accuracy of MNIST image classification. Number of trials $=20$, $K=10$, $\gamma=0.01$, $x_\text{max}=-x_\text{min}=2$, number of epochs $=4$, training samples per device $=5000$, unique digit class per device $=1$, batch size $=32$. $M=2, L=4, \beta=10^4$. \textit{Not} conditional on $|g_{m,k}|$.}\label{fig:FL}
\end{figure}

\section{Examples of framework coverage}\label{sec:demonstration_coverage}

As already summarized in Table \ref{tab:mapping_applications}, our analysis framework covers many NC-OAC schemes proposed in the literature. In what follows, we discuss some of these in more detail:

\begin{itemize}
	\item \textbf{Jam signalling}: In \cite{krohn2005collaborative} an RFID reader should retrieve a histogram of ``sell-by-dates'' of ``goods'' equipped with RFID tags. To this end, each day of the year is allocated 4 slots where the RFID transmits non-coherently in one of the 4 slots randomly, corresponding to their date. The scanner detects (counts) whether at least one transmission has occured in one of the slots for each day, then constructs a histogram based on detections in $365\times 4$ slots. For NC-OAC this is performed by devices transmitting $x_k\in\{0,1\}$ with the Affine mapping on one of $365\times 4$ codewords, where the decoder is a binary detector.
	\item \textbf{Balanced number systems}: In \cite{csahin2023over}, source data $x_k\in\mathbb{R}$ corresponding to FL gradient elements is encoded to a set of $D$ integers (``numerals''). Each numeral is taken from a set of $L$ integers. The ``balanced number system'' works since summation of the numerals from each device is equivalent to the wanted sum $x$. Each possible value for each numeral is mapped to a unique codeword, giving $D\times L$ codewords in total. Each device transmits 1 on the codewords corresponding to the value of each numeral, and 0 on the others. At the receiver the NC-OAC sum of each codeword gives the sum of the numerals across the devices, which can be decoded to the wanted sum $x$. 
	\item \textbf{Chirp-based majority vote}: In \cite{hoque2023chirp} FL is done with binary gradient elements ($x_k\in\{-1,1\}$) and majority vote at the receiver. The NC-OAC used is majority vote with AA as described herein, the codewords are modulated using chirps.  
	\item \textbf{Decentralized consensus}: In \cite{michelusi2024non} a decentralized consensus process is used to aggregate model updates for decentralized gradient descent. Consider $\beta_{k,k'}=\beta_{k',k}$ the channel power between device $k$ and $k'$, where each device knows only the sum of its downlink channel powers $\sum_{k'\neq k}\beta_{k,k'}$. Furthermore, let $\rho_k=P$ and $\eta=1/P$. During each iteration of the consensus process, the devices are split into two alternating groups, one transmitting and one receiving. The transmitting group projects data onto $[-1,1]$ and exploits the Augmented Affine mapping. In expectation the receivers measure $\sum_{k'\neq k}\beta_{k',k}w_{k'}$ (to some scalar factor compensating for neighbors not transmitting) which is a weighted (biased) sum of the values from the neighbors of $k$.  While this is a biased sum, decentralized consensus as in \cite{michelusi2024non} only requires the consensus weights ($\beta_{k,k'}$ and self-weight) to sum to zero by using $\sum_{k'\neq k}\beta_{k,k'}$. Furthermore, while \cite{michelusi2024non} uses $2d+1$ codewords (compared to $2d$ as herein) to aggregate $d$ scalars, the final codeword functions as a pilot as it enforces the codewords to sum to 1, enabling instantaneous estimation of $\sum_{k'\neq k}\beta_{k,k'}$.
\end{itemize}

\section{Conclusion}

We developed a framework for the analysis of  Non-Coherent Over-the-Air Computation (NC-OAC) for  unbiased estimation of the sum of distributed source data,  when channel phase knowledge is unavailable at the access point and the devices. 
We particularly studied two of the most common NC-OAC methods of mapping real-valued source data to non-negative codewords. Then, we showed that a mapping that splits the dynamic range of source data across two resources can decrease the sum estimation variance by an order of magnitude, compared to a mapping that repeats the entire range over two resources. In addition to sum estimation, our framework supports discrete counting and majority voting across participating devices. Finally, we propose and study a new mapping that exhibits superior performance compared to existing mappings for sum estimation. 

Possible directions for future work are to: (1) find a closed-form expression for the MSE of the projected codeword sum estimator in (\ref{eq:projection}),  (2) further investigate optimal majority vote and count detection, and 
(3) prove the conjecture that the optimal $N$ for the proposed Extended Affine mapping is $N=2$ or $N=4$. Finally, there are many more variations of mappings to study, such as finding the optimal Augmented Affine mapping, and mappings robust to outliers. One may also mix the Affine and Augmented Affine mappings in the same aggregation instance for $L\geq 3$: One channel use for the Affine, two for the Augmented Affine. Another alternative mapping is switching the sign of $a$ in the Unit Affine mapping described in Example \ref{simple_A}, which would horizontally mirror the Affine mapping curve in Fig. \ref{fig:A_vs_AA}. Such ways of mixing the mappings could enable further adapting $\var(\widehat{x}|\mathbf{x})$ to the distribution of $x$. 
\appendices

\section{Proof of Theorem \ref{theorem_cond_var_w}}\label{proof_theorem_cond_var_w}	
	The proof of Theorem \ref{theorem_cond_var_w} is split in two parts, first for statistical CSI, then for instantaneous CSI.
\begin{proof}[Proof Theorem \ref{theorem_cond_var_w} for Statistical CSI]
	First it is shown that 
	\begin{equation}\label{eq:unbiased_equality}
		\eta w_k\widehat{\beta_k}\rho_k = w_k, \forall k \Leftrightarrow \mathbb{E}[\widehat{w}|\mathbf{W}]=w,
	\end{equation}
	must hold: It follows from (\ref{theestimator}), (\ref{eq:SC_cross_correlation}), (\ref{eq:second_moment_gaussian_r}),  (\ref{eq:IC_crosscorr}), (\ref{eq:IC_secondmoment}) that  
	\begin{equation}\label{eq:unbiased_equality2}
		\begin{split}
			&\mathbb{E}\left[\widehat{w}|\mathbf{W}\right] = w \Leftrightarrow \eta\sum_{k=1}^{K} w_{k}\widehat{\beta_k}\rho_k = \sum_{k=1}^{K} w_{k}\Leftrightarrow\\
			&\sum_{k=1}^{K}(\eta w_{k}\widehat{\beta_k}\rho_k - w_k) = 0.
		\end{split}
	\end{equation}
	The last equality in (\ref{eq:unbiased_equality2}) is only satisfied under one of two conditions: Either 
	\begin{equation}\label{eq:equal_weight}
		\eta w_k\widehat{\beta_k}\rho_k - w_k=0, \forall k,
	\end{equation}
	or for at least one $k$, say $k'$
	\begin{equation}
		\eta w_{k'}\widehat{\beta_{k'}}\rho_{k'} - w_{k'}<0,
	\end{equation}
	implying that $\eta w_{k'}\widehat{\beta_{k'}}\rho_{k'} = \alpha_{k'} w_{k'}, \alpha_{k'} < 1$. However, assume all devices except ${k'}$ transmit $w_k=0$, then $w=w_{k'}$, but the AP will measure $\alpha_{k'} w_{k'}\neq w$. This means knowledge of $w$ is required either at ${k'}$ or at the AP such that the receive scaling factor $\eta$ can be adjusted, but this is not possible by the practical nature of the problem. Thus only (\ref{eq:equal_weight}) can hold, and is therefore both sufficient and necessary to make (\ref{theestimator}) unbiased. This shows that (\ref{eq:unbiased_equality}) must hold. Now, select $\rho_k$ for an equally weighted sum of codewords at the receiver:
	\begin{equation}\label{eq:rho_statistical_CSI}
		\eta w_k \rho_k \beta_k = w_k \Leftrightarrow \rho_k = \frac{1}{\eta\beta_k},
	\end{equation}
	which by (\ref{eq:unbiased_equality}) makes the codeword sum estimator (\ref{theestimator}) unbiased. Thus $\rho_k$ has been derived and $\mathbb{E}[\widehat{w}|\mathbf{W}]=w$ proven. Next, consider the conditional variance of $\widehat{w}$, where $(*)$ follows from tedious but standard calculations:
	\begin{equation}\label{w_variance_statistical_CSI}
		\begin{split}
			&\text{var}(\widehat{w}|\mathbf{W}) = \var\left(\frac{\widetilde{\mathbf{r}}^H\widetilde{\mathbf{r}}}{ML}-\eta\bigg|\mathbf{W}\right) = \frac{1}{M^2L^2}\var\left(\widetilde{\mathbf{r}}^H\widetilde{\mathbf{r}}|\mathbf{W}\right)\\
			&\overset{(*)}{=}\frac{1}{ML}\left((w + \eta)^2 + (L-1) \sum_{k=1}^Kw_k^2\right).
		\end{split}
	\end{equation}
	From (\ref{w_variance_statistical_CSI}) it is apparent that the variance is minimized by minimizing $\eta$. When the device with the weakest channel $\beta_k$, say $\beta_\text{min}$, transmits $w_\text{max}$ (\ref{eq:rho_statistical_CSI}) should still hold. Thus, the minimum $\eta$ is given when this device transmits at maximum power $P$:
	\begin{equation}\label{eq:statistical_csi_eta}
		\eta w_\text{max}\rho_k\beta_\text{min} = \eta P \beta_\text{min} = w_\text{max} \Leftrightarrow \eta = \frac{w_\text{max}}{P\beta_\text{min}}.
	\end{equation}
	This concludes the proof of Theorem \ref{theorem_cond_var_w} for the statistical CSI.
\end{proof} 
\begin{proof}[Proof Theorem 2 for Instantaneous CSI]
	Using the codeword sum estimator (\ref{theestimator}) conditional on $|g_{m,k}|$, one gets
	\begin{equation*}
		\begin{split}
			\mathbb{E}\left[\widehat{w}|\mathbf{W}\right] &= \frac{\eta\sum_{m=1}^M\sum_{k=1}^Kw_k|g_{m,k}|^2\rho_k}{M} = \eta\sum_{k=1}^Kw_k\rho_k\widehat{\beta_k}, 
		\end{split}
	\end{equation*}
	such that (\ref{eq:unbiased_equality}), like for statistical CSI in (\ref{eq:rho_statistical_CSI}) gives $\rho_k = \frac{1}{\eta \widehat{\beta_k}}$.
	This shows $\rho_k$ for the case of instantaneous CSI. Now consider $\var(\widehat{w}|\mathbf{W})$ when $M=1$, since then $\widetilde{r}_{m,l}\equiv \widetilde{r}_l$ are i.i.d. across $l$ given $|g_{k}|$ and $\mathbf{W}$. Then $\widetilde{r}_l^*\widetilde{r}_l$ are also i.i.d. across $l$ which gives:
	\begin{equation*}
		\begin{split}
			\var\left(\widehat{w}|\mathbf{W}\right)& = \var\left(\frac{\widetilde{\mathbf{r}}^H\widetilde{\mathbf{r}}}{L}-\eta\bigg|\mathbf{W}\right) = \frac{1}{L^2}\var\left(\widetilde{\mathbf{r}}^H\widetilde{\mathbf{r}}|\mathbf{W}\right)\\
			&=\frac{1}{L}\var\left(\widetilde{r}_{1,1}^*\widetilde{r}_{1,1}|\mathbf{W}\right)\equiv \frac{\eta^2}{L}\var\left(r^*r|\mathbf{W}\right),
		\end{split}
	\end{equation*}
	where $r\equiv r_{1,1}$ inserted into (\ref{eq:rml}) gives
	\begin{equation*}
		\begin{split}
			\var\left(r^*r|\mathbf{W}\right) &= \frac{1}{\eta^2}\left(w^2-\sum_{k=1}^Kw_k^2\right)+ \frac{2}{\eta}\sum_{k=1}^{K}w_k + 1,
		\end{split}
	\end{equation*}
	which if multiplied by $\eta^2$ gives (\ref{eq:varwgivenw}) for IC, which is minimized by minimizing $\eta$. Then through the same arguments as for statistical CSI in (\ref{eq:statistical_csi_eta}) one gets $\eta = \frac{w_\text{max}}{P\widehat{\beta_\text{min}}}$.
\end{proof}
 
\section{Proof of Theorem \ref{theorem_cond_variance}}\label{proof_theorem_cond_variance}
\begin{proof}[Proof]
	Recall (\ref{xestimator}) and (\ref{eq:varwgivenw}). The proof is split into four parts, one for each combination of CSI and mapping. Note that $\cdot|\mathbf{x}\Leftrightarrow \cdot|\mathbf{W}$ since the relationship is deterministic.
	\subsubsection{Statistical CSI for Unit Affine mapping}
	Consider coefficients $a,b$ generally, and recall from (\ref{eq:A_constraints}) that these fully determine $c,d$. Then, the conditional variance is found by insertion of  (\ref{A_encoder}), (\ref{A_decoder}) into (\ref{w_variance_statistical_CSI}):
	\begin{equation}\label{condvarSA}
		\begin{split}
			&\var(\widehat{x}|\mathbf{x})_{SC,A} = \var(c\widehat{w} - d\mathbf{x}) = c^2\var\left(\widehat{w}|\mathbf{W}\right)\\
			&\overset{(\ref{eq:simple_affine_mapping_coeff})}{=}\frac{1}{ML}\left(\left(x + K +2\eta\right)^2 + (L-1)\sum_{k=1}^K\left(x_k+1\right)^2\right).
		\end{split}
	\end{equation}
	
	\subsubsection{Statistical CSI for Unit Augmented Affine mapping}\label{statistical_CSI_AA_cond_var}
	Let $w_1=x_+ \equiv \sum_{k=1}^K(x_k)^+$, ${w_2=x_-\equiv\sum_{k=1}^{K}(-x_k)^+}$, such that $x=x_+-x_-$, then for $i\in\{1,2\}$
	\begin{equation}\label{eq:wcondvartox}
		\begin{split}
			\var(\widehat{w_i}|\mathbf{x}) & = \frac{2}{ML}\left(\left(w_i+\eta\right)^2+\left(\frac{L-2}{2}\right)\sum_{k=1}^{K}w_{k,i}^2\right),
		\end{split}
	\end{equation}
	where we have a sequence length of $L/2$ per codeword, and 
	\begin{equation*}
		\begin{split}
			&w_1 = \sum_{k=1}^K\underbrace{w_{k,1}}_{=(x_k)^+} ,\text{ }w_2=\sum_{k=1}^K\underbrace{w_{k,2}}_{=(-x_k)^+}.
		\end{split}
	\end{equation*}
	Given $\mathbf{x}$, the estimates $\widehat{w_1},\widehat{w_2},$ are mutually independent as they are aggregated separately. Even if the aggregation is sequential (as for mapping $A$) causing a statistical dependence between the channels for $w_1$ and $w_2$, a single device only contributes to one of $w_1$ or $w_2$. Then, assuming channels across devices are independent, $w_1$ and $w_2$ are always independent given $\mathbf{x}$. Importantly, $AA$ exploits the same number of independent channel instances as $A$ in both the sequential and separate aggregation cases. The independence results in
	\begin{equation}\label{condvarSAA}
		\begin{split}
			&\var(\widehat{x}|\mathbf{x})_{SC,AA} = \var\left(\widehat{w_1}-\widehat{w_2}|\mathbf{x}\right)=\sum_{i=1}^{2}\var\left(\widehat{w_i}|\mathbf{x}\right)\\
			&\overset{(\ref{eq:wcondvartox})}{=}\frac{2}{ML}\left(\left(x_++\eta\right)^2 + \left(x_- + \eta\right)^2+\left(\frac{L-2}{2}\right)\sum_{k=1}^{K}x_{k}^2\right).
		\end{split}
	\end{equation}
	
	\subsubsection{Instantaneous CSI for Unit Affine mapping}
	As for statistical CSI, consider general $a,b$. Insertion of (\ref{A_encoder}) and (\ref{A_decoder}) into (\ref{eq:varwgivenw}) for IC gives:
	\begin{equation}\label{condvarIA}
		\begin{split}
			&\var(\widehat{x}|\mathbf{x})_{IC,A} = \var(c\widehat{w} +d|\mathbf{x}) = c^2\var\left(\widehat{w}|\mathbf{W}\right)\\
			&= \frac{1}{a^2L}\Bigg(\left(ax+bK + \eta\right)^2 - \sum_{k=1}^K\left(ax_k+b\right)^2\Bigg)\\
			&\overset{(\ref{eq:simple_affine_mapping_coeff})}{=}\frac{1}{L}\left(\left(x+K + 2\eta\right)^2 - \sum_{k=1}^K\left(x_k+1\right)^2\right).
		\end{split}
	\end{equation}
	\subsubsection{Instantaneous CSI for Unit Augmented Affine mapping}
	As previously argued in Section \ref{statistical_CSI_AA_cond_var} for the case of statistical CSI with the Augmented Affine mapping, one gets
	\begin{equation}
		\begin{split}
			&\var(\widehat{w_i}|\mathbf{x}) = \frac{2}{L}\left(\left(w_i + \eta\right)^2 - \sum_{k=1}^Kw_{k,i}^2\right),
		\end{split}
	\end{equation}
	and
	\begin{equation}\label{condvarIAA}
		\begin{split}
			&\var(\widehat{x}|\mathbf{x})_{IC,AA} = \var\left(\widehat{w_1}-\widehat{w_2}|\mathbf{x}\right)=\sum_{i=1}^{2}\var\left(\widehat{w_i}|\mathbf{x}\right)\\
			&\overset{(\ref{simplemappings})}{=}\frac{2}{L}\left(\left(x_+ + \eta\right)^2 + \left(x_- + \eta\right)^2 - \sum_{k=1}^K{x_k}^2\right).
		\end{split}
	\end{equation}
\end{proof}

\section{Proof of Theorem \ref{theorem_total_variance}}\label{proof_theorem_total_variance}
\begin{proof}[Proof]
	First, recall the law of total variance (\ref{eq:total_variance}) and let $x_k\sim\text{U}(-1, 1)$, i.i.d. across $k$, with ${\mathbb{E}[x_k] = 0, \var(x) = K/3}$. Furthermore, $\eta$ is used throughout for notational ease, which can be replaced with the normalized $\eta_\text{AA}$ for AA as a final step ensuring equal energy expenditure.
	
	\subsection{Unit Affine Mapping}
	Given $\mathbb{E}\left[\widehat{\mathbf{w}}|\mathbf{W}\right]=\mathbf{w}$, for general $a,b$ which by (\ref{eq:A_constraints}) fully determines $c,d$ it follows that
	\begin{equation*}
			\mathbb{E}[\widehat{x}|\mathbf{x}] = c\mathbb{E}\left[\widehat{w}|\mathbf{W}\right] + d = cw + d = x + cKb - Kcb = x,
	\end{equation*}
	such that $\var(\mathbb{E}[\widehat{x}|\mathbf{x}])=\var(x)$. Furthermore, since $x_k\sim\text{U}(-1,1)$ w.l.o.g let $a>0$. Then since $w_k = ax_k + b\geq 0$ it follows that $b\geq a$ and $w_\text{max}=a+b$ such that $\eta = \frac{a+b}{P\beta_\text{min}}$. Then from (\ref{condvarSA}) we get
	\begin{equation}\label{varSA1}
		\begin{split}
			&\left(\var(\widehat{x})_{\text{SC,A}}-\frac{K}{3}\right)a^2ML\\
			& = \mathbb{E}\left[\left(\left(ax + Kb + \eta\right)^2+ (L-1)\sum_{k=1}^K\left(ax_k+b\right)^2\right)\right]\\
			& =\Bigg(\frac{Ka^2}{3} + \left(Kb+\frac{a+b}{P\beta_\text{min}}\right)^2 + (L-1)K\left(\frac{a^2}{3}+b^2\right)\Bigg).
		\end{split}
	\end{equation}
	Now, from (\ref{varSA1}) it is clear that $\var(\widehat{x})_{\text{SC,A}}$ decreases with decreasing $b$ such that $b=a$ is the optimal $b$. Then finally, with $a=1/2$ from (\ref{eq:simple_affine_mapping_coeff}) it follows that 	
	\begin{equation*}
		\begin{split}
			&\var(\widehat{x})_{\text{SC,A}}=\frac{1}{ML}\left(K^2 - K + 4K\eta + 4\eta^2 + \frac{4KL}{3}\right) + \frac{K}{3},\\
		\end{split}
	\end{equation*}
	showing the minimum $\var(\widehat{x})$ for the Affine mapping, achieved by the Unit Affine mapping. For the case of instantaneous CSI, recall (\ref{condvarIA}) and the arguments around $b\geq a$, $w_\text{max}=a+b$ from the case of statistical CSI. Then
	\begin{equation}\label{varIA}
		\begin{split}
			&\left(\var(\widehat{x})_{\text{IC,A}}-\frac{K}{3}\right)a^2L = \mathbb{E}\Bigg[\left(ax+bK + \eta\right)^2\\
			& - \sum_{k=1}^K\left(ax_k+b\right)^2\Bigg] = \left(b^2K(K-1) + 2bK\eta + \eta^2\right).\\
		\end{split}
	\end{equation}
	Again, similar to (\ref{varSA1}) it is clear that (\ref{varIA}) is minimized by setting $b=a$, giving $\eta=\frac{2a}{P|g_\text{min}|}$ and
	\begin{equation*}
		\begin{split}
			&\var(\widehat{x})_{\text{IC,A}}\overset{a\text{ from }(\ref{eq:simple_affine_mapping_coeff})}{=} \frac{1}{L}\left(K^2 - K + 4K\eta + 4\eta^2\right) + \frac{K}{3},
		\end{split}
	\end{equation*}
	which concludes the proof for the Affine mapping, where the minimum $\var(\widehat{x})$ has been derived and proven.
	\subsection{Unit Augmented Affine Mapping}
	As for the Unit Affine mapping, given  $\mathbb{E}\left[\widehat{\mathbf{w}}|\mathbf{W}\right]=\mathbf{w}$, it follows for the Unit Augmented Affine mapping from (\ref{eq:simple_augmented_affine_coeff}) that
	\begin{equation*}
			\mathbb{E}[\widehat{x}|\mathbf{x}] = \mathbb{E}\left[\widehat{w_1}|\mathbf{W}\right] - \mathbb{E}\left[\widehat{w_2}|\mathbf{W}\right] = w_1 - w_2 = x,
	\end{equation*}
	such that $\var(\mathbb{E}[\widehat{x}|\mathbf{x}])=\var(x)$. Now, let $\widetilde{x_k}\sim\text{U}(0,1)$, ${q_k\sim\text{Be}(1/2)}$, i.i.d. across $k$, with moments ${\mathbb{E}[\widetilde{x_k}]=1/2}$, ${\var(\widetilde{x_k})=1/12}$, ${\mathbb{E}[q_k] = 1/2}$, ${\var(q_k) = 1/4}$.  Equivalently we have $x_k=q_k\widetilde{x_k} -(1-q_k)\widetilde{x_k}\sim\text{U}(-1, 1)$, ${(x_k)^+ = q_k\widetilde{x_k}}$, ${(-x_k)^+=(1-q_k)\widetilde{x_k}}$. Then for statistical CSI recall (\ref{condvarSAA}), giving
	\begin{equation*}
		\begin{split}
		&\var(\widehat{x})_{\text{SC,AA}}-\frac{K}{3}\\
		&=\frac{2}{ML}\mathbb{E}\left[\left(x_++\eta\right)^2 + \left(x_- + \eta\right)^2+\left(\frac{L-2}{2}\right)\sum_{k=1}^{K}x_{k}^2\right]\\
		&=\frac{1}{ML}\left(\frac{K^2}{4} + \frac{5K}{12} + 2K\eta + 4\eta^2+\frac{K(L-2)}{3}\right).\\
		\end{split}
	\end{equation*}
	Finally, for instantaneous CSI recall (\ref{condvarIAA}), giving
	\begin{equation}\label{varIAA}
		\begin{split}
			&\var(\widehat{x})_{\text{IC,AA}} = \frac{2}{L}\mathbb{E}\left[2x_+^2 + 4\eta x_+ + 2\eta^2 - \sum_{k=1}^Kx_k^2\right]\\
			&+ \var(x)= \frac{1}{L}\left(\frac{K^2}{4} - \frac{K}{4}+ 2K\eta + 4\eta^2\right) + \frac{K}{3}.
		\end{split}
	\end{equation}
\end{proof}

\section{Proof of Theorem \ref{theorem:generalized_affine_mapping}}\label{proof_generalized_affine_mapping}
For the estimation of $\mathbf{w}$, excluding $K_{1}$ and $K_{N/2+1}$, the codeword sum estimator (\ref{theestimator}) with Theorem \ref{theorem_cond_var_w} directly yields
\begin{equation}\label{eq:extended_affine_cond_variances}
	\begin{split}
		&\var(\widehat{K_n}|\mathbf{W}) = \frac{1}{ML_b}\left(\left(K_n + \eta\right)^2 + (L_b-1)K_n\right),\\
		&\var(\widehat{w}_n|\mathbf{W}) = \frac{1}{ML_w}\left(\left(w_n + \eta\right)^2 + (L_w-1) \sum_{k=1}^Kw_{k,n}^2\right).
	\end{split}
\end{equation} 
Since these codewords are estimated on independent resources, their estimates are independent given $\mathbf{W}$ (equivalent to given $\mathbf{x}$), which means the variance of the sum $\widehat{x}=\mathcal{D}(\mathbf{\widehat{w}})$ can be separated into sum of variances:
\begin{equation}\label{eq:sum_of_variances_extended_affine}
	\begin{split}
		&\var(\widehat{x}|\mathbf{x})\\
		& =\frac{4}{N^2}\bigg(\sum_{n=1}^{N}\var\left(\widehat{w_{n}}|\mathbf{x}\right)+ \sum_{n=1}^{N/2}\var(\widehat{K_n}|\mathbf{x})(n-1)^2\\
		& + \sum_{n'=N/2+1}^{N}\var(\widehat{K_n'}|\mathbf{x})\left(n'-\frac{N}{2}-1\right)^2\bigg).
	\end{split}
\end{equation}
To compute the total variance (\ref{eq:total_variance}), consider $x_k\sim\text{U}(-1,1)$ and that $\var(\mathbb{E}[\widehat{x}|\mathbf{x}])=\var(x)=K/3$. Then $w_{k,n}$ is a mixture of $0$ and $\text{U}(0,1)$ mixed by $b_{k,n}$: $\mathbb{E}[w_{k,n}] = \frac{1}{2N}$, $\mathbb{E}[w_{k,n}^2]=\frac{1}{3N}$, $\mathbb{E}[K_n] = K/N$. We get
\begin{equation}
	\begin{split}
		&\mathbb{E}[(w_{n}+\eta)^2] = \frac{K}{3N} + \frac{K(K-1)}{4N^2} + \frac{K\eta}{N} + \eta^2,\\
		&\mathbb{E}[(K_n + \eta)^2] = \frac{K}{N} + \frac{K(K-1)}{N^2} + \frac{2\eta K}{N} + \eta^2,
	\end{split}
\end{equation}
which inserted into (\ref{eq:extended_affine_cond_variances}) and then (\ref{eq:sum_of_variances_extended_affine}) gives $\mathbb{E}[\var(\widehat{x}|\mathbf{x})]$. Finally,  $\mathbb{E}[\var(\widehat{x}|\mathbf{x})]+K/3$ and $\eta$ replaced with $\eta_N$, gives Theorem \ref{theorem:generalized_affine_mapping}.

\section{Proof of Proposition \ref{prop:fair_transmission}}\label{proof_fair_transmission}
\begin{proof}
	For $x_k\in\text{U}(-1,1)$ the expected power across all $L$ transmitted symbols, for the Unit Affine mapping is
	\begin{equation*}
		\begin{split}
			&\mathbb{E}\left[\sum_{l=1}^{L}|t_{k,l}|^2\right]=L\mathbb{E}\left[w_k\rho_k\right] = L\rho_k\mathbb{E}\left[\frac{x+1}{2}\right] = \rho_k\frac{L}{2},
		\end{split}
	\end{equation*}
	and for the Extended Affine mapping
	\begin{equation*}
		\begin{split}
			&\mathbb{E}\left[\sum_{l=1}^{L}|t_{k,l}|^2\right]=\sum_{n=1}^{N}\mathbb{E}\left[\text{I}(x_k\in\mathcal{X}_n)\left(L_ww_{k,n}\rho_k +\dots\right.\right.\\
			&\left.\left.L_b\rho_k\text{I}\left(x_k\notin\left\{1,\frac{N}{2}+1\right\}\right)\right)\right]\\
			&=\frac{1}{N}\left[(N-2)\left(\frac{L_w\rho_k}{2} + L_b\rho_k\right) + 2\frac{L_w\rho_k}{2}\right]\\
			&=\rho_k\left(\frac{L+L_b(N-2)}{2N}\right).
		\end{split}
	\end{equation*}
	To normalize for energy with the Unit Affine as reference, let $\rho_{k,N}$ be the power control coefficient for the Extended Affine mapping with $N$ segments, and $\eta_N$ the corresponding receive coefficient. The result on normalization is derived by equating the two sums above. Finally, $\eta_N<\eta$ follows from $NL > L + L_b(N-2)\Leftrightarrow L(N-1)>L_b(N-2)$ since $L=L_wN + L_b(N-2)$. 
\end{proof}

\section{Federated Learning with Unbiased Aggregation}\label{sec:FL_mse}
In FL, $K$ devices collaborate to solve a learning task coordinated by a central server. The goal is to find the optimal model
\begin{equation}\label{eq:FL_objective}
	\boldsymbol{\theta}^* = \underset{\boldsymbol{\theta}}{\text{ arg min }}F(\boldsymbol{\theta}) \equiv \frac{1}{K}\sum_{k=1}^{K}f_k(\boldsymbol{\theta}), 
\end{equation}
where $\boldsymbol{\theta}\in\mathbb{R}^D$ are \textit{global parameters} and $f_k(\cdot)$ \textit{local objectives} which can have heterogeneous local optima $\boldsymbol{\theta}_k^*\neq\boldsymbol{\theta}_{k'}^*\neq\boldsymbol{\theta}^*$. For example, the local objectives $f_k(\cdot)$ can be $D$-dimensional neural networks with equal architecture and loss function, but with different \textit{local data} at each device. In FL, (\ref{eq:FL_objective}) is solved through iterations of stochastic gradient descent (SGD). First, in iteration $t\in\mathbb{N}$, all devices compute their local (stochastic) gradient $\nabla f_k(\boldsymbol{\theta}^t)\in\mathbb{R}^D$, then transmit this to the central server which updates $\boldsymbol{\theta}^t$ as follows
\begin{equation}\label{eq:FL_SGD_STEP}
	\boldsymbol{\theta}^{t+1}=\boldsymbol{\theta}^{t} - \gamma\frac{1}{K}\sum_{k=1}^{K}\nabla f_k(\boldsymbol{\theta}^t),
\end{equation}
for some step-size $\gamma$. Finally, $\boldsymbol{\theta}^{t+1}$ is broadcast to all devices, and the FL loop continues \cite{yang2020federated}. 
\subsection{The Case of Unbiased Aggregation for Federated Learning}
If $f_k(\cdot)$ are strongly convex and smooth and $\gamma$ is sufficiently small, inserting an unbiased,
finite-variance estimator of $\nabla f_k(\cdot)$ into (\ref{eq:FL_SGD_STEP}) guarantees convergence of $\boldsymbol{\theta}^t$ to an $\mathcal{O}(\gamma)$-neighborhood of the global optimum $\boldsymbol{\theta}^*$, in the mean-square sense; see, e.g., \cite{larsson2025unified} or \cite{gower2019sgd}. In contrast, with biased gradient estimates, convergence is generally not guaranteed. For example, in an extremely biased setting where a single out of $K$ devices  is exclusively weighted at the receiver, the global model will converge to the local optimum of that single device. In the next section, we demonstrate from first principles, under the above-specified circumstances, how gradient unbiasedness naturally enters as a condition for convergence of FL. Analyses of the impact of gradient bias, under more general conditions, can be found for example in \cite{abrar2024biased} and \cite{demidovich2023guide}.
\subsection{Convergence in Mean-Square With Unbiased Aggregation}
Let each local objective $f_k(\cdot)$ be differentiable, $L$-smooth and $\mu$-strongly convex. Then the global optimal point of $F(\cdot)$,  $\boldsymbol{\theta}^*$, is unique. 
Suppose  $\mathbf{g}^t$ is an unbiased, finite-variance estimate of $\nabla F(\cdot)$ such that $\mathbf{g}^t = \nabla F\left(\boldsymbol{\theta}^t\right) + \mathbf{e}^t$, where
\begin{equation}\label{eq:noise_stats}
	\begin{split}
		&\mathbb{E}\left[\mathbf{e}^t|\boldsymbol{\theta}^t\right]=0,\text{ } \mathbb{E}\left[\big|\big|\mathbf{e}^t\big|\big|^2\big|\boldsymbol{\theta}^t\right]\leq V,
	\end{split}
\end{equation}
for some constant $V$. Then the FL SGD step is
\begin{equation}
	\boldsymbol{\theta}^{t+1} = \boldsymbol{\theta}^{t} -  \gamma\nabla F\left(\boldsymbol{\theta}^t\right) - \gamma\mathbf{e}^t, 
\end{equation}
where for sufficiently small step-size $\gamma$, the mean-square errors can be upper bounded as follows:

\begin{equation}
	\begin{split}
		&\mathbb{E}\left[\big|\big|\boldsymbol{\theta}^{t+1} - \boldsymbol{\theta}^*\big|\big|^2\right]\\
		& = \mathbb{E}\left[\mathbb{E}\left[\big|\big|\boldsymbol{\theta}^{t} - \boldsymbol{\theta}^* -  \gamma\nabla F\left(\boldsymbol{\theta}^{t}\right) -\gamma \mathbf{e}^t\big|\big|^2\bigg|\boldsymbol{\theta}^t\right]\right]\\
		& \overset{(*)}{=} \mathbb{E}\left[\mathbb{E}\left[\big|\big|\left(\boldsymbol{\theta}^{t} - \boldsymbol{\theta}^*\right) -  \gamma \mathbf{A}\left(\boldsymbol{\theta}^{t} - \boldsymbol{\theta}^*\right) -\gamma \mathbf{e}^t\big|\big|^2\bigg|\boldsymbol{\theta}^t\right]\right]\\
		& \overset{(**)}{=}\mathbb{E}\left[\big|\big|\left(\mathbf{I} - \gamma \mathbf{A}\right)\left(\boldsymbol{\theta}^t - \boldsymbol{\theta}^*\right)\big|\big|^2\right] + \gamma^2\mathbb{E}\left[\big|\big|\mathbf{e}^t\big|\big|^2\big|\boldsymbol{\theta}^t\right]\\
		&\leq (1-\gamma\mu)^2\mathbb{E}\left[\big|\big|\boldsymbol{\theta}^{t} - \boldsymbol{\theta}^*\big|\big|^2\right] + \gamma^2V\leq\dots\\
		&\leq(1-\gamma\mu)^{2t}\mathbb{E}\left[\big|\big|\boldsymbol{\theta}^{1} - \boldsymbol{\theta}^*\big|\big|^2\right] + \gamma\frac{V}{\mu}\left(\frac{1-(1-\gamma\mu)^{2t}}{2-\gamma\mu}\right),
	\end{split}
\end{equation}
converging to $\mathcal{O}(\gamma)$ when $t\rightarrow \infty$. Here $(*)$ follows from the \textit{Mean Hessian Theorem} \cite{larsson2025unified}, which states that with $f_k(\cdot)$ as herein, for any $\mathbf{x}, \mathbf{y}\in\mathbb{R}^D \text{ one can find }\text{ }\mathbf{A}\in\mathbb{R}^{D\times D}, \text{ s.t. }L\mathbf{I}\succeq\mathbf{A}\succeq\mu\mathbf{I}\text{ and }\nabla F(\mathbf{y}) - \nabla F(\mathbf{x}) = \mathbf{A}(\mathbf{y} - \mathbf{x})$. Furthermore, $(**)$ follows from (\ref{eq:noise_stats}) and only holds if the estimates are unbiased ($\mathbf{e}$ is zero-mean). The first inequality exploits that the norm is sub-multiplicative, and that the spectral norm of $\mathbf{I}-\gamma\mathbf{A}$ is at most $1-\gamma\mu$ for sufficiently small $\gamma$. The second inequality follows from recursion. 

\section{Discussion on the complexity of NC-OAC}\label{seq:complexity}
First consider the complexity of channel estimation overhead given by Table \ref{tab:complexity}. Assume a channel coherence time $\tau> 2K$ of the instantaneous CSI $g_k$, furthermore, that the large-scale fading $\beta_k$ is constant over $T$ coherence periods ($T\times\tau$ time-steps). Estimation of $g_k, \beta_k, \forall k$ require $K$ uplink pilots followed by feedback to each device. The table below shows the channel estimation overhead per aggregated codeword over a $T\times\tau$ time horizon when $M=L=1$:

\begin{table}[!t]
	\begin{center}
		\renewcommand{\arraystretch}{1.5} 
		\setlength{\tabcolsep}{10pt}      
		\caption{Comparison of overhead complexity.}
		\label{tab:complexity}
		\begin{tabular}{ | c |c | } 
			\hline
			OAC type & Average overhead per codeword \\ 
			\hline
			Coherent & $\frac{2K}{2(\tau-2K)}$ \\ 
			\hline
			Non-coherent IC & $\frac{2K}{\tau-2K}$ \\ 
			\hline
			Non-coherent SC & $\frac{2K}{\tau T-2K}$ \\ 
			\hline
		\end{tabular}
	\end{center}
\end{table}
The overhead for both coherent OAC, and NC-OAC with IC is higher than for NC-OAC with SC. However, NC-OAC with SC needs $L > 1$ to reach a comparable estimation variance. In terms of computational complexity at the AP side, the OAC schemes require $\mathcal{O}(ML)$ operations, which is better than the $\mathcal{O}(MLK)$ without OAC.

	\bibliographystyle{IEEEtran}
	
	\bibliography{references.bib}
		
	\begin{IEEEbiography}[{\includegraphics[width=1in,height=1.25in,clip,keepaspectratio]{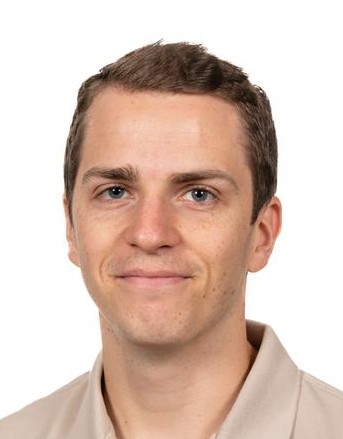}}]{Martin Dahl}
		received the B.Sc. and M.Sc. degrees in Electrical Engineering from Linköping University, Sweden, in 2021 and 2023, respectively, where he is currently pursuing the Ph.D. degree in Electrical Engineering. During his M.Sc. studies, he spent two semesters at EPFL, Switzerland. His M.Sc. thesis received multiple recognitions and he has completed internships at Ericsson Research and Goldman Sachs. Martin is currently affiliated with the Wallenberg AI, Autonomous Systems and Software Program (WASP).
	\end{IEEEbiography}
	\begin{IEEEbiography}[{\includegraphics[width=1in,height=1.25in,clip,keepaspectratio]{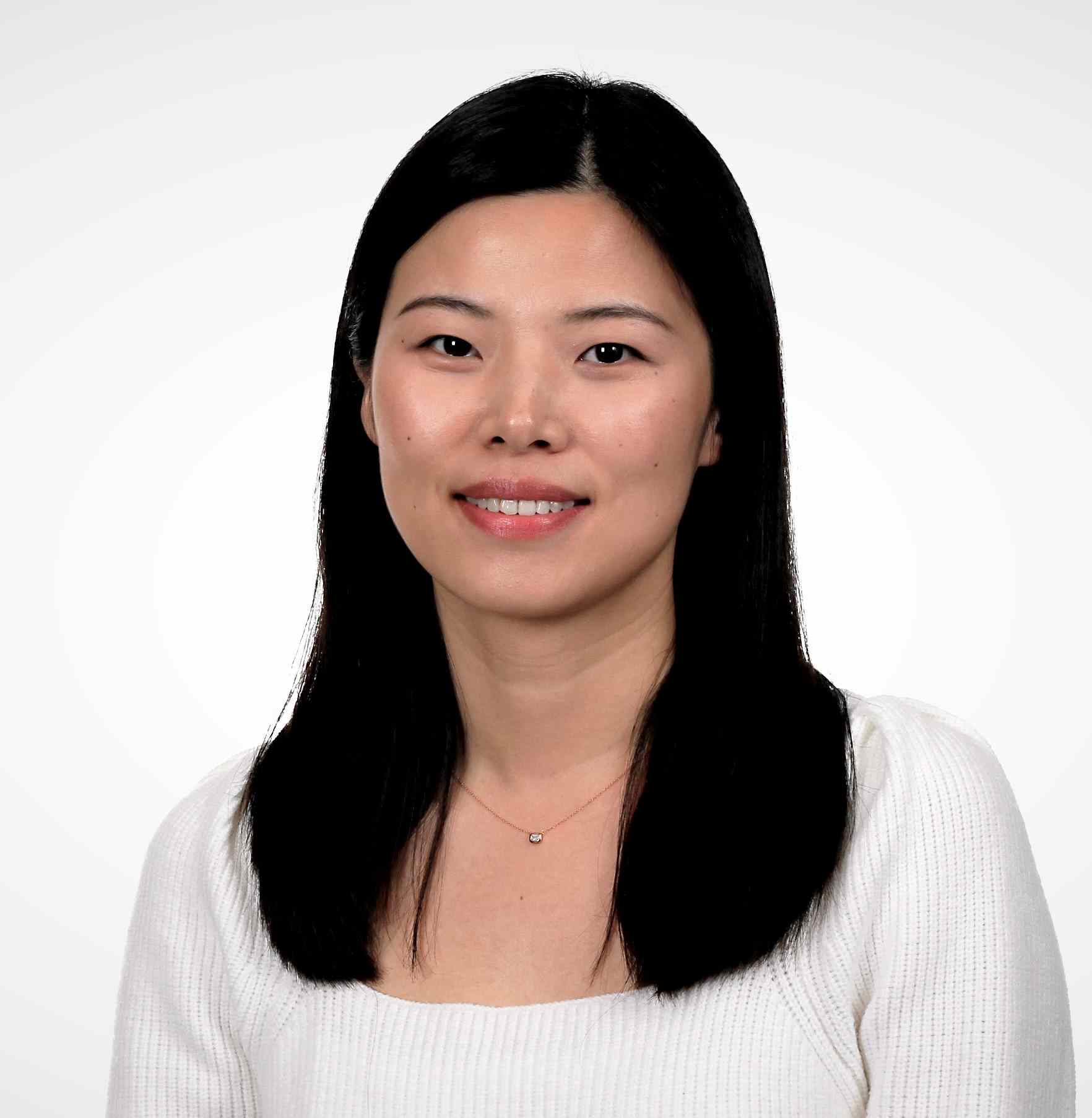}}]{Zheng Chen}
		is an Associate Professor with the Department of Electrical Engineering at Linköping University, Sweden. She received her M.Sc. and Ph.D. degrees from CentraleSupélec, Université Paris-Saclay, France, in 2013 and 2017, respectively. Her research focuses on distributed and cooperative computing, optimization, and learning for networked intelligent systems under communication constraints.
		
		She was a recipient of the 2020 IEEE Communications Society Young Author Best Paper Award. She has served as the co-chair of several workshops and special sessions at IEEE GLOBECOM, ICASSP, SPAWC, Asilomar, and as the technical program chair of the 2022 IEEE SPS-EURASIP summer school on “Defining 6G: Theory, Applications and Enabling Technologies”. She is currently an Associate Editor of the IEEE Transactions on Wireless Communications, IEEE Transactions on Communications, and IEEE Transactions on Green Communications and Networking. She is also serving as the Lead Guest Editor for IEEE JSAC special issue on “Distributed Optimization, Learning, and Inference over Communication-Constrained Networks”
	\end{IEEEbiography}
	\begin{IEEEbiography}[{\includegraphics[width=1in,height=1.25in,clip,keepaspectratio]{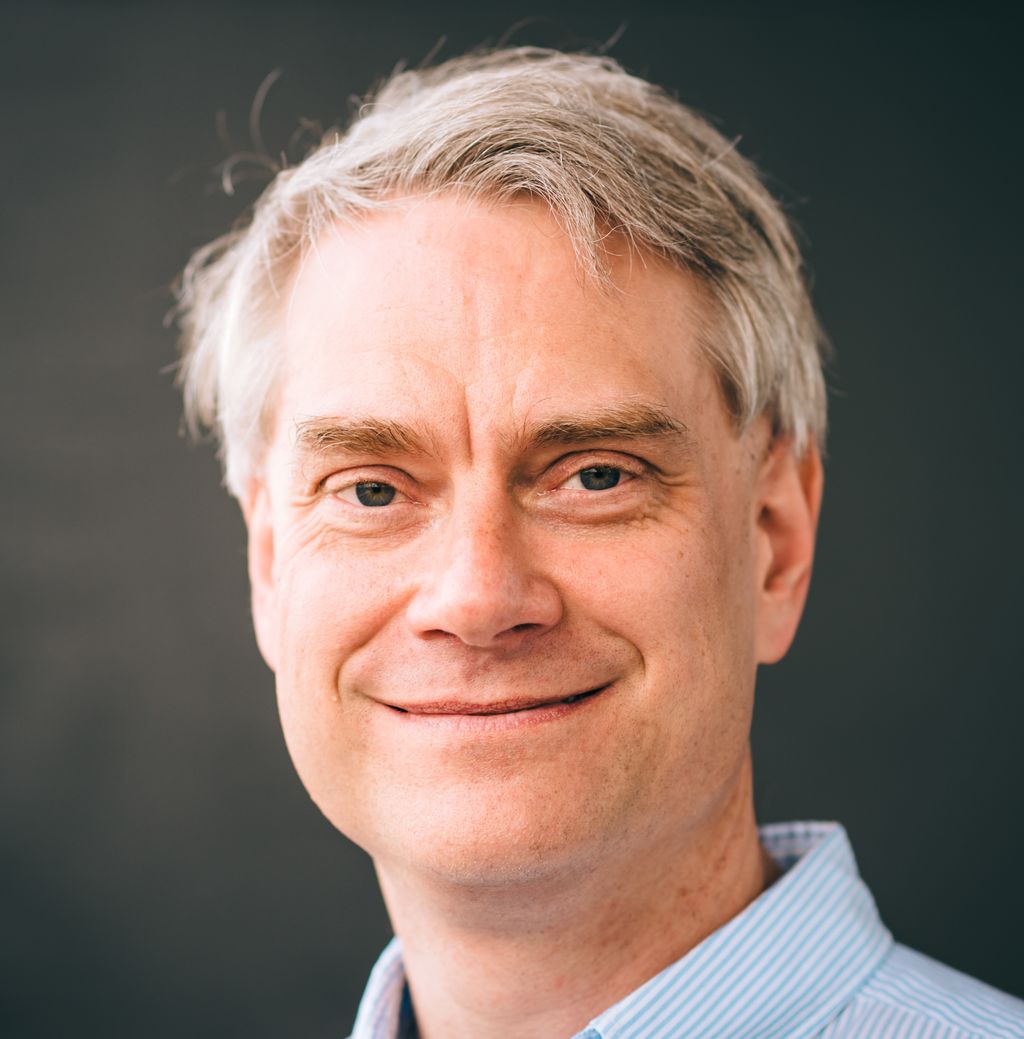}}]{Erik G. Larsson}
	received the Ph.D. degree from Uppsala University,
	Uppsala, Sweden, in 2002.  He is currently Professor of Communication
	Systems at Link\"oping University (LiU) in Link\"oping, Sweden. He was
	with the KTH Royal Institute of Technology in Stockholm, Sweden, the
	George Washington University, USA, the University of Florida, USA, and
	Ericsson Research, Sweden.  His main professional interests are within
	wireless communications, signal processing, network science, and decentralized machine learning. He
	co-authored \emph{Space-Time Block Coding for  Wireless Communications} (Cambridge University Press, 2003) 
	and \emph{Fundamentals of Massive MIMO} (Cambridge University Press, 2016). 
	
	He served as  chair  of the IEEE Signal Processing Society SPCOM technical committee (2015--2016), 
	chair of  the \emph{IEEE Wireless  Communications Letters} steering committee (2014--2015), 
	member of the  \emph{IEEE Transactions on Wireless Communications}    steering committee (2019-2022),
	General and Technical Chair of the Asilomar SSC conference (2015, 2012), 
	technical co-chair of the IEEE Communication Theory Workshop (2019), 
	and   member of the  IEEE Signal Processing Society Awards Board (2017--2019).
	He was Associate Editor for, among others, the \emph{IEEE Transactions on Communications} (2010-2014), 
	the \emph{IEEE Transactions on Signal Processing} (2006-2010),
	and  the \emph{IEEE Signal  Processing Magazine} (2018-2022).
	
	He received the IEEE Signal Processing Magazine Best Column Award
	twice, in 2012 and 2014, the IEEE ComSoc Stephen O. Rice Prize in
	Communications Theory in 2015, the IEEE ComSoc Leonard G. Abraham
	Prize in 2017, the IEEE ComSoc Best Tutorial Paper Award in 2018, 
	the IEEE ComSoc Fred W. Ellersick Prize in 2019, and the
	IEEE SPS Donald G. Fink Overview Paper Award in 2023.
	
	He is a Fellow of the IEEE, a Fellow of EURASIP, a member of the  Royal Swedish  Academy of Sciences (KVA), 
	a member of the  Royal Swedish Academy of Engineering Sciences (IVA), and Highly Cited according to ISI Web of Science.
	
	\end{IEEEbiography}

\end{document}